\documentclass[11pt,letterpaper]{article}
\usepackage{amsmath,amssymb,amsthm}
\usepackage{fullpage}
\usepackage{changepage}
\usepackage{enumitem}
\usepackage{complexityToC}
\usepackage{float} 

\usepackage[pdfstartview=FitH,
            pdfpagemode=UseNone,
            backref,
            colorlinks=true]{hyperref}
            \hypersetup{urlcolor=[rgb]{0.8, .2, .5},
                         linkcolor=[rgb]{0.1,0.4,0.6},
                         citecolor=[rgb]{0.4,0.3,0.7}
            }
\usepackage[capitalize,noabbrev,nameinlink]{cleveref}
 \usepackage{tikz} 
\date{}
\newtheorem{theorem}{Theorem}[section]
\newtheorem{lemma}[theorem]{Lemma}
\newtheorem{corollary}[theorem]{Corollary}
\newtheorem{proposition}[theorem]{Proposition}

\newtheorem{definition}[theorem]{Definition}
\newtheorem*{theorem*}{Theorem}
\newtheorem*{lemma*}{Lemma}
\newtheorem*{corollary*}{Corollary}
\newtheorem*{proposition*}{Proposition}
\newtheorem*{claim*}{Claim}
\newtheorem*{definition*}{Definition}
\newtheorem*{open problem*}{Open problem}
\newtheorem{open problem}[theorem]{Open problem}
\newtheorem*{conjecture}{Conjecture}
\newtheorem*{algorithm*}{Algorithm}

\newtheorem{stretchalgorithm}[theorem]{Demi-Bit Stretching Algorithm}
\newtheorem{summary}[theorem]{Summary}
\theoremstyle{remark}
\newtheorem*{comment*}{Comment}

\newtheorem*{remark*}{Remark}

\newcommand\blankpage{
  \null
  \thispagestyle{empty}%
  \addtocounter{page}{-1}%
  \newpage} 

\newcommand{\pr}[1]{\mathbb{P}\left[{#1}\right]}
\renewcommand{\Pr}[2][]{\underset{#1}{\mathbb{P}}[{#2}]}
\newcommand{\ppoly}{\ensuremath{\P/\sf{poly}}}
\newcommand{\nppoly}{\ensuremath{\NP/\sf{poly}}}
\newcommand{\conppoly}{\ensuremath{\coNP/\sf{poly}}}
\newcommand{\bi}[1]{\{0, 1\}^{#1}}
\newcommand{\tildenppoly}{\ensuremath{\sf{N}\tilde{\P}/qpoly}}
\newcommand{\smallpoly}{\ensuremath{\sf{poly}}}
\newcommand*\circled[1]{\tikz[baseline=(char.base)]{
  \node[shape=circle,draw,inner sep=1pt] (char) {#1};}}

\newcommand{\para}{\paragraph}

\newcommand{\bN}{\mathbb{N}}

\newcommand{\bP}{\mathbb{P}}

\newcommand{\bR}{\mathbb{R}}

\newcommand{\McA}{\mathcal{A}}

\newcommand{\McD}{\mathcal{D}}

\renewcommand{\epsilon}{\varepsilon}
\renewcommand{\phi}{\varphi}

\DeclareSymbolFont{sfoperators}{OT1}{cmss}{m}{n}
\DeclareSymbolFontAlphabet{\mathsf}{sfoperators}

\makeatletter
\def\operator@font{\mathgroup\symsfoperators}
\makeatother
\newcommand{\triplenorm}[1]{{\vert\kern-0.25ex\vert\kern-0.25ex\vert #1
        \vert\kern-0.25ex\vert\kern-0.25ex\vert}}
\newcommand{\striplenorm}[1]{{\left\vert\kern-0.25ex\left\vert\kern-0.25ex\left\vert #1
        \right\vert\kern-0.25ex\right\vert\kern-0.25ex\right\vert}}

\newcommand{\verteq}{\rotatebox{90}{$\,=$}}

\newcommand{\bits}{\ensuremath{\{0,1\}}}
\newcommand{\N}{\ensuremath{\mathbb{N}}}

\newcommand{\kpoly}{\ensuremath{\mathrm{K}^{\poly}}}   
\newcommand{\kpolys}[1]{\ensuremath{\kpoly[#1]}}  

\newcommand{\kolm}{\ensuremath{\mathrm{K}}}   
\newcommand{\kolmt}{\ensuremath{\mathrm{K}^t}}   
\newcommand{\kolmts}[1]{\ensuremath{\kolmt[#1]}} 
\newcommand{\kolmtns}[1]{\ensuremath{\kolm^{t(n)}[#1]}}

\author{Iddo Tzameret\thanks{Department of Computing. Part of this
work was done at Simons Institute for the Theory of Computing, UC Berkeley. Part of this project has received funding from the European Research Council (ERC) under the European Union’s Horizon 2020 research and innovation programme (grant agreement No.~101002742). Email: \text{iddo.tzameret@gmail.com}} \smallskip 
\\ 
 \small Imperial College London
\and Lu-Ming Zhang\thanks{Part of this work was done while at Imperial College London. Email:~\text{l534zhang@uwaterloo.ca}} 
\smallskip
\\
\small London School of Economics and \\
\small Political Science
}

\title{Stretching Demi-Bits and Nondeterministic-Secure Pseudorandomness}

\usepackage{color}
\usepackage[normalem]{ulem}
\usepackage{graphicx}

\usepackage[usenames,dvipsnames]{pstricks}
\usepackage{pstricks-add}
\usepackage{epsfig}
\usepackage{pst-grad} 
\usepackage{pst-plot} 
\usepackage[space]{grffile} 
\usepackage{etoolbox} 

\begin{document}


\clearpage\maketitle\thispagestyle{empty}


\begin{abstract}
%

We develop the theory of  cryptographic nondeterministic-secure pseudorandomness beyond the point reached by Rudich's original work \cite{superbit}, and  apply it to draw new consequences in average-case complexity and proof complexity. Specifically, we show the following:

%

 
\begin{description}[leftmargin=6pt]
\item[Demi-bit stretch:] Super-bits and demi-bits are variants of cryptographic pseudorandom generators which are secure against nondeterministic statistical tests \cite{superbit}. They were introduced to  rule out  certain approaches to proving strong complexity lower bounds beyond the limitations set out by the Natural Proofs barrier of Rudich and Razborov \cite{natural}.  Whether  demi-bits are stretchable  at all had been an open problem since their introduction. We answer this question affirmatively by showing that:
every demi-bit $b:\bits^n\to \bits^{n+1}$ can be stretched 
   into sublinear many demi-bits $b':\bits^{n}\to \bits^{n+n^{c}}$, for every constant $0<c<1$.

\item[Average-case hardness:] Using work by Santhanam \cite{San20}, we apply our results to obtain new  average-case Kolmogorov complexity results:
we show that $\kpolys{n-O(1)}$ is zero-error average-case hard against \NP/\poly\ machines iff $\kpolys{n-o(n)}$ is, where  for a function $s(n):\N\to\N$, $\kpolys{s(n)}$ denotes the languages of all strings $x\in\bits^n$ for which there are  (fixed) poly-time Turing machines of description-length at most $s(n)$ that output $x$.



\item[Characterising super-bits by nondeterministic unpredictability:] In the deterministic setting, Yao \cite{Yao1} proved that  super-polynomial hardness of pseudorandom generators is equivalent to (``next-bit'') \emph{unpredictability}. Unpredictability roughly means that given any strict prefix of a random string, it is infeasible to predict the next bit. We initiate the study of unpredictability beyond the deterministic setting (in the cryptographic regime), and characterise the nondeterministic hardness of  generators  from an unpredictability perspective.
Specifically, we propose four stronger notions of unpredictability:
$\NP/\smallpoly$-unpredictability,
$\coNP/\smallpoly$-unpredictability, 
$\cap$-unpredictability and $\cup$-unpredictability, 
and show that super-polynomial nondeterministic hardness of generators lies  between \mbox{$\cap$-unpredictability} and  $\cup$-unpredictability.

\item[Characterising super-bits by nondeterministic hard-core predicates:] We introduce a nondeterministic variant of hard-core predicates, called \emph{super-core predicates}.  We show that the existence of a super-bit is equivalent to the existence of a super-core of some non-shrinking function. This serves as an analogue  of the equivalence between the existence of a strong pseudorandom generator and the existence of a hard-core of some one-way function \cite{hard_core_from_OWF, PRG_from_OWF}, and provides a first alternative characterisation of super-bits. We also prove that a certain class of functions, which may have hard-cores, cannot possess any super-core.


\end{description}
\end{abstract}



\setcounter{tocdepth}{3}
\tableofcontents


%
%
%
%
%
%
%

\setcounter{page}{2}

\section{Introduction} \label{ch1}

Pseudorandomness is a natural concept allowing to measure the extent to which a resource-bounded computational machine can identify  true random sources. It is an important notion in algorithms,  enabling to derandomize efficient probabilistic algorithms by simulating many true random bits using fewer random bits. Another important aspect of pseudorandomness lies in computational complexity and cryptography, and specifically computational lower bounds, where it serves as the foundation of many results in  cryptography,   hardness vs.~randomness trade-offs, and several barriers to proving strong computational lower bounds. Here, we will  mostly be interested in the latter aspect of barrier results.

\subsection{The theory of  nondeterministic-secure pseudorandomness}


Razborov and Rudich established a connection between pseudorandomness and the ability to prove boolean circuit lower bounds in their Natural Proofs paper \cite{natural}. They showed that most lower bounds arguments in circuit complexity  contain (possibly implicitly) an efficient  algorithm for deciding hardness, in the following sense: given the truth table of a boolean function, the algorithm determines  if the function possesses some (combinatorial) properties that imply it is hard for a certain given circuit class (they called this ``constructivity" and ``usefulness" of a lower bound argument). They moreover showed that this algorithm  identifies correctly the hardness of a non-negligible fraction of functions (which is called ``largeness'' in \cite{natural}). On the other hand, such an efficient algorithm for determining the  hardness of boolean functions (for general circuits) would contradict reasonable assumptions in pseudorandomness, namely the existence of strong pseudorandom generators. This puts a \emph{barrier}, so to speak, against the ability to prove lower bounds using natural proofs.

The notion of natural proofs has had a great influence on computational complexity theory.  However, the fact that the plausible nonexistence of certain classes of natural proofs  provides an obstacle against \emph{very constructive} lower bound arguments (namely, those arguments that implicitly contain an efficient algorithm to determine when a function is hard) is somewhat less desirable. One would hope to extend the obstacle to less constructive proofs, for instance, proofs whose arguments contain implicitly only short \emph{witnesses} for the hardness of functions. 

Indeed, the notion of natural proofs comes to explain the difficulty in \emph{proving} lower bounds, not in efficiently \emph{deciding} hardness of boolean functions. For these reasons, among others, Rudich \cite{superbit} set to extend the natural proofs barrier so that they encompass non-constructive arguments, namely, arguments implicitly using efficient \emph{witnesses} of the hardness of boolean functions, in contrast to deterministic  algorithms.
This was done by extending the notion of pseudorandom generators so that they are secure against \emph{nondeterministic} adversaries. 

While empirically most known lower bound proofs were shown, at least implicitly, to fall within the scope of \ppoly-constructive natural proofs, it is definitely conceivable and natural to assume that some lower bound approaches will  necessitate \NP/\poly-constructive natural proofs. In fact, it is a very interesting open problem in itself to find such an  \NP/\poly-constructive lower bound proof method. Furthermore, in  works in proof complexity the role of nondeterministic-secure pseudorandomness is  important (cf.~recent work by Pich and Santhanam~\cite{PS19}, as well as Kraj\'{i}\v{c}ek~\cite{Kra04}). Also, note that when dealing with the notion of barriers we refer to impossibly results and so it cannot always be expected to come up with good examples of proof methods we wish to rule out.

Accordingly, to extend the natural proofs barrier, Rudich introduced two primitives: \emph{super-bits,} and its weaker variant \emph{demi-bits}. Super-bits and demi-bits are (non-uniform) nondeterministic variants of strong pseudorandom generators (PRGs). They are secure against nondeterministic adversaries (i.e., adversaries in \nppoly{}).
Demi-bits require their nondeterministic adversaries to break them in a \emph{stronger} sense than super-bits, and hence their existence  constitutes a better (i.e., \emph{weaker}) assumption than the existence of super-bits on which to base barrier results; and this is one reason why the concept of demi-bits is important.

 More specifically, super-bits and demi-bits both require a nondeterministic adversary to meaningfully distinguish truly random strings from pseudorandom ones by certifying truly random ones (i.e., an adversary outputs $1$ if it thinks a given string is an output of a truly  random process and $0$ if it is the output of  a pseudorandom generator). Thus, a nondeterministic distinguisher cannot break super-bits nor demi-bits by simply guessing a seed of the generator. Precisely, this is guaranteed as follows: for demi-bits we insist that strings in the image of the pseudorandom generator are always \emph{rejected}, while for super-bits we allow some such pseudorandom strings  to be accepted but we insist that many  more strings outside the image of the pseudorandom generator are accepted (than strings in the image of the pseudorandom generator).

Formally, we have the following (all the distributions we consider in this work are, by default, uniform, unless otherwise stated, and $U_n$ denotes the uniform distribution over $\{0, 1\}^n)$:

\begin{definition}[Nondeterministic hardness \cite{superbit}]\label{def:nondeterministic-hardness}
Let $g_n: \{0, 1\}^n \to \{0, 1\}^{l(n)}$, with $l(n)>n$, be a function in $\P/\poly$. We call such a function a \textbf{generator}. The \textbf{nondeterministic hardness} $H_{\textup{nh}}(g_n)$  (also called \textbf{super-hardness}) of  $g_n$ is the minimal $s$ for which there exists a nondeterministic circuit $D$ of size at most $s$ such that
\begin{equation}\label{eq:ov:tag-1}
    \underset{y \in \{0, 1\}^{l(n)}}{\mathbb{P}} [D(y) = 1] - \underset{x \in \{0, 1\}^n}{\mathbb{P}} [D(g_n(x)) = 1] \geq \frac{1}{s}.
\end{equation}

\end{definition}
In contrast to the standard definition of (deterministic) hardness (\Cref{def:standard-hardness}), the order of the two possibilities on the left-hand side  is crucial. This order forces a nondeterministic distinguisher to certify the randomness of a given input. Reversing the order, or adding an absolute value to left-hand side, trivializes (as in standard  PRGs)  the task of breaking $g$: a distinguisher $D$ can simply guess a seed $x$ and check if $g(x)$ equals the given input. For such a $D$, we have
$\pr{D(g(x))=1} = 1$ and $\pr{D(y) = 1} \leq 1/2$.

Super-bits are  exponentially super-hard generators:
\begin{definition}[Super-bits \cite{superbit}]
A generator $g: \{0, 1\}^n \to \{0, 1\}^{n + c}$ (computable in \P/\poly), for some $c: \bN \to \bN$, is called $c$ \textbf{super-bit(s)} (or a $c$-super-bit(s)) if $H_{\textup{nh}}(g) \geq 2^{n^\epsilon}$ for some constant $\epsilon > 0$ and all sufficiently large $n$'s. In particular, if $c = 1$, we call $g$ a super-bit.
\end{definition}

Many candidates of strong PRGs (against deterministic machines) were constructed by exploiting functions conjectured to be one-way and/or their hard-cores (see for example~\cite{factoring, discrete-logarithm, subset_sum}). The work \cite{factoring} presented PRGs based on the assumption that factoring is hard, while \cite{discrete-logarithm} presented PRGs based on the conjectured hardness of the discrete-logarithm problem. \cite{subset_sum} presented PRGs based on the subset sum problem.
Similarly, for nondeterministic-secure PRGs (namely, super-bits), Rudich conjectured the existence of a super-bit generator based on the hardness of the subset sum problem \cite{subset_sum}. We are unaware of any additional conjectured construction of super-bits. 
\smallskip 

As mentioned above, another  hardness measure of generators is introduced by Rudich:

\begin{definition}[Demi-hardness  \cite{superbit}]
Let $g_n: \{0, 1\}^n \to \{0, 1\}^{l(n)}$ be  a generator (computable in \ppoly). Then the \textbf{demi-hardness} $H_{\textup{dh}}(g_n)$ of  $g_n$ is the minimal $s$ for which there exists a nondeterministic circuit $D$ of size at most $s$ such that
\begin{equation}\label{eq:ov:tag-2}
    \underset{y \in \{0, 1\}^{l(n)}}{\mathbb{P}} [D(y) = 1] \geq \frac{1}{s}
    \text{\quad and \ }
    \underset{x \in \{0, 1\}^n}{\mathbb{P}} [D(g_n(x)) = 1] = 0.
\end{equation}
\end{definition}We note that \eqref{eq:ov:tag-2}, which requires a distinguisher to make no mistake on generated strings, is a stronger requirement than \eqref{eq:ov:tag-1}. Thus,
$H_{\textup{nh}}(g) 
\leq
H_{\textup{dh}}(g)$ for every generator $g$.

 Demi-bits are exponentially demi-hard  generators, where ``demi" here  stands for  ``half'':
\begin{definition}[Demi-bits  \cite{superbit}]\label{def:overv-demi-bit}
$g: \{0, 1\}^n \to \{0, 1\}^{n + c}$ for some $c: \bN \to \bN$ is called $c$ \textbf{demi-bit(s)} (or a $c$-demi-bit(s)) if $H_{\textup{dh}}(g) \geq 2^{n^\epsilon}$ for some $\epsilon > 0$ and all sufficiently large $n$'s. In particular, if $c = 1$, we call $g$ a demi-bit.
\end{definition}

It is worth mentioning that  demi-bits $g$ as in \Cref{def:overv-demi-bit} can be viewed as a \emph{hitting set generator} against \NP/\poly\ (see below \Cref{sec:overv-apps} and Santhanam \cite{San20}).

The difference between super-bits and demi-bits is that demi-bits require their distinguishers to break them in a stronger sense: a demi-bit distinguisher must always be correct on the pseudorandom strings (i.e., always output $0$ for strings in the image of the  generator). Thus, if $g$ is  a super-bit(s) (the plural here denotes that $g$ may have a stretching-length greater than  $1$; if the stretching length is exactly $1$, we say $g$ is \emph{a} super-bit), no algorithm in \nppoly{} can break $g$ in the weaker sense \eqref{eq:ov:tag-1}, and hence no algorithm in \nppoly{} can break $g$ in the stronger sense \eqref{eq:ov:tag-2}, which means $g$ is also a demi-bit(s). In other words, the existence of super-bits implies the existence of demi-bits (although it is open if any of these two exists).

\begin{remark*}[Cryptographic vs.~complexity-theoretic regime] 
In this work, we are interested only in the \emph{cryptographic} regime of pseudorandomness. In this regime, the adversary whom the generator tries to fool is allowed to be stronger than the generator and specifically has sufficient computational resources to run the generator. In the \emph{complexity-theoretic regime}, in which the adversary cannot simulate the generator, the notion of nondeterministic secure pseudorandomness was developed in works by, e.g.,~Klivans and van Melkebeek \cite{KvM02} as well as  Shaltiel and Umans \cite{SU05} (see also the recent work by Sdroievski and van Melkebeek \cite{SvM23} and references therein). These complexity-theoretic ideas have also found several applications in cryptography (originating from the work of Barak, Ong and Vadhan \cite{BOV07}).
It is also worth mentioning that in the complexity-theoretic regime, one can use the original definition of the hardness of PRGs (\Cref{def:standard-hardness}) even against nondeterministic adversaries; while this is not the case in the cryptographic regime, in which a PRG as in \Cref{def:standard-hardness} can never be safe against a nondeterministic adversary who guesses the seed.



\end{remark*}


\subsection{Relations to barrier results}

%
%
%
%

Recall that natural proofs \cite{natural} for proving circuit lower bounds are proofs that use a natural combinatorial property of boolean functions. A combinatorial property (or \emph{a property}, for short) $C$ of boolean functions is a set of boolean functions. We say a function $f$ has  property $C$ if $f$ is in $C$. Let $\Gamma \text{ and } \Lambda$ be complexity classes, and $F_n$ be the set of all $f : \bi{n} \to \bi{}$. We say $C$ is $\Gamma$-natural if a subset  $C'\subseteq C$  satisfies \textit{constructivity}, that is, it is $\Gamma$-decidable whether $f$ is in $ C'$, and \textit{largeness}, that is, $C'$ constitutes a non-negligible portion of $F_n$. We say $C$ is \textit{useful} against  $\Lambda$ if every function family $f$ that has property $C$ infinitely often is not computable in $\Lambda$. The idea of natural proofs is that, if we want to prove some function family $f$ (e.g., the boolean satisfiability problem SAT) is not in \ppoly{} (or in general, some other complexity class $\Lambda$), we identify some natural combinatorial property $C$ of $f$ and show all function families that have property $C$ are not in \ppoly{} (i.e., the property is useful against \ppoly{}). If $f$ is $\NP$-complete (e.g., SAT), then such a proof concludes $\P \neq \NP$.


Razborov and Rudich argued that, based on the existence  of strong PRGs, no \ppoly-natural proofs can be useful against \ppoly{}. They showed that many known proofs of lower bounds against (non-monotone) boolean circuits are natural or can be presented as natural in some way. 

In this context, the theory of nondeterministic-secure generators allows one to rule out a larger  (arguably more natural) class of lower bound arguments: 
\begin{theorem}[\cite{superbit}]\label{thm:overview-Rudich}
If super-bits exist, then there are no $N\tilde{P}/qpoly$-natural properties useful against \ppoly{}, where $N\tilde{P}/qpoly$ is the class of languages recognised by non-uniform, quasi-polynomial-size circuit families.
\end{theorem}

This theorem  is proved based on the ability to \emph{stretch} super-bits, namely, taking a generator that maps $n$ bits to $n+1$ bits, which we refer to as stretching length 1, to a generator that maps $n$ bits to $n+N$ bits, with $N>1$, which we call stretching length $N$. In the standard theory of pseudorandomness, a hard-bit (i.e., a strong PRG with stretching length $1$) is shown to be stretchable to polynomially many hard-bits (i.e., a polynomial stretching-length) \cite{Blum_Micali, Yao1} and can  be exploited to construct hard-to-break pseudorandom \emph{function} generators (loosely speaking, generators that generate pseudorandom functions indistinguishable from truly random ones) \cite{random_function}. As a hard-bit, a super-bit can also be stretched, using similar stretching algorithms, to polynomially many super-bits and to pseudorandom function generators secure against nondeterministic adversaries \cite{superbit}. The proofs of the correctness of such stretching algorithms are based on a technique called the \emph{hybrid argument} \cite{hybrid} reviewed below. In contrast, whether a demi-bit can be stretched even to two demi-bits was  unknown before the current work, since this cannot be concluded with a direct application of a  standard hybrid argument.


\section{Contributions, significance and context}
We develop the foundations of nondeterministic-secure pseudorandomness.
This is the first systematic investigation into nondeterministic pseudorandomness (in
the cryptographic regime) we are aware of, building on the primitives proposed by Rudich \cite{superbit}.
We provide new understanding of the primitives of the theory,  namely,  super and demi-bits, as well as introducing new notions and showing how they relate to established ones. We draw several conclusions from these results in average-case and proof complexity. We also  achieve some modest progress on establishing sounder foundations for barrier results: by showing, for instance, that demi-bits can be (moderately) stretched,  we provide some hope to strengthen the connection between demi-bits and unprovability results (as of now, it is only known that the existence of a super-bit yield barrier results, while we hope to show that the weaker assumption of the existence of demi-bits suffices for that matter).   
%

\subsection{Stretching demi-bits}\label{sec:intro:stretc-demi-bit}
In \Cref{alg:stretch}, we provide an algorithm that achieves a sublinear stretch for any given demi-bit. This solves the open problem of whether a demi-bit can be stretched to 2-bits \cite[Open Problem 2]{superbit} (see also Santhanam 
\cite[Question 4]{San20}).

\begin{theorem*}[Informal; \autoref{thm:alg}]
Every demi-bit $b:\{0,1\}^n\to\{0,1\}^{n+1}$ can be efficiently converted (stretched) into demi-bits $g:\{0,1\}^n\to\{0,1\}^{n+n^c}$, for every constant $0<c<1$.
\end{theorem*}


\subsubsection{Discussion and significance of stretching demi-bits to barrier results}

 Stretching demi-bits can be viewed as a first step towards showing that the existence of a demi-bit rules out $N\tilde{P}/qpoly$-natural properties useful against \ppoly{}, as we explain below. Providing such a barrier for  $N\tilde{P}/qpoly$-natural properties based on the existence of a demi-bit is important, since assuming the existence of a demi-bit is a weaker assumption than assuming the existence of a super-bit. 

Why is stretching demi-bits a step towards showing that the existence of a demi-bit would rule out $N\tilde{P}/qpoly$-natural properties useful against \P/\poly? The reason is that stretching is the first step in the argument to base barrier results on the existence of a super-bit, in the following sense:  the existence of a super-bit implies barrier results because one can stretch super-bits to obtain pseudorandom function generators, from which one gets the barrier result as noted in   \autoref{thm:overview-Rudich} above (and the text that
follows it). 
More precisely, stretching demi-bits is a first (and necessary) step towards Rudich's Open Problem 3, and this problem also implies Rudich's Open Problem 4:

\begin{open problem*}[Rudich's Open Problem 3 \cite{superbit}]
Given a  demi-bit,  is it possible to build a pseudorandom function generator with exponential ($2^{n^\epsilon}$) demi-hardness?
\end{open problem*}

\begin{open problem*}[Rudich's Open Problem 4 \cite{superbit}]Does the existence of  demi-bits rule out   $N\tilde{P}/qpoly$-natural properties useful against \ppoly{}? \end{open problem*}


Moreover, the study of the stretchability of demi-bit(s)  provides a perspective towards resolving Rudich's Open Problem 1 (a positive answer of which would also resolve positively Open Problem 4):

\begin{open problem*}[Rudich's Open Problem 1 \cite{superbit}]
Does the existence of a demi-bit imply the existence of a super-bit?
\end{open problem*}

This is because the stretchability of super-bits is well understood, while previously we did not know anything about the stretchability of demi-bits. We  expect that understanding better  basic properties of  demi-bits, such as stretchability, would shed light on the relation between the existence of demi-bits and the existence of super-bits
(Open problem 1). 




\subsubsection{Applications in average-case complexity}\label{sec:overv-apps} 
Here we describe an application of \Cref{thm:alg} to the average-case hardness of time-bounded Kolmogorov complexity. 

As observed by Santhanam \cite{San20},  a hitting set generator $g: \{0,1\}^n \rightarrow \{0,1\}^{n+1}$
 exists iff there exists a demi-bit $b: \{0,1\}^n \rightarrow \{0,1\}^{n+1}$. To recall, a \textbf{\emph{hitting set generator} \emph{against a  class 
of decision problems}} $\mathcal{C}\subseteq 2^{\{0,1\}^{N}}$ is a function
$g: \{0,1\}^n \rightarrow \{0,1\}^{N}$, for $n<N$, such that the image of $g$ \emph{hits} (namely, intersects) every dense enough set $A$ in $\mathcal{C}$ (that is,  $|A|\ge \frac{2^{N}}{N^{{O(1)}}}$). And we have:

\begin{proposition*}[\Cref{prop:santh-claim}; \cite{San20}]
Let $n<N$. A hitting set generator $g: \{0,1\}^n \rightarrow \{0,1\}^{N}$ computable in the class $\mathcal D$ against $\NP/\poly$
exists iff there exists a demi-bit $b: \{0,1\}^n \rightarrow \{0,1\}^{N}$ computable in $\mathcal D$ (against $\NP/\poly$).
\end{proposition*}  


Santhanam \cite[Proposition 3]{San20} established an equivalence between (succinct) hitting set generators and average-case hardness of MCSP, where MCSP stands for  the \textit{minimal circuit size problem}. However, as mentioned to us by Santhanam \cite{San22}, similar arguments can show an equivalence between hitting set generators (not-necessarily succinct ones) and poly-time bounded Kolmogorov complexity zero-error average-case hardness against \NP/\poly\ machines, as we show in this work.

 
\smallskip 


We define the \textbf{$t$-\textit{bounded Kolmogorov complexity of string $x$}}, denoted $\kolm^t(x)$, to be the minimal length of a string $D$ such that the universal Turing machine $U(D)$ (we fix some such universal machine) runs in time at most $t$ and outputs $x$.
See \cite{All2006} for more details about time-bounded Kolmogorov complexity and Definition 9 there for the definition  of time-bounded Kolmogorov complexity of strings  (that definition actually produces the $i$th bit of the string $x$ given an index $i$ and $D$ as inputs to $U$, but this does not change our result).




\newcommand{\kpolyp}{\ensuremath{\kpoly[s(n)]}} 

\begin{definition}[the language \kolmts{s} and \kpolyp]\label{def:kpoly}
For a time function $t(n):\N\to\N$ and a size function $s(n):\N\to\N$, such that $s(n)\le n$, let $\kolmtns{s(n)}$ be the language $\{ x\in\bits^*\;:\; |X|=n ~\land ~ \kolm^{t(n)}(x)\le s(n) \}$.
We define \kpoly$[s(n)]$ to be the language 
$\bigcup\nolimits_{c\in\N}\kolm^{n^c}[s(n)]$. 
\end{definition}

  

We also need to define the concept of zero-error average-case hardness against the class \NP/\poly\ (see \Cref{def:zeac-hardness}).
Informally, for a language $L$ to be zero-error average-case \emph{easy} for \NP/\poly, there should be a nondeterministic polytime machine with advice such that given an input $x$ the machine guesses a witness for $x\in L$ or a witness for $x\not\in L$, and when the witness is found it answers accordingly; and moreover we assume that for a polynomial-small fraction of inputs there are such witnesses (for membership or non-membership in $L$). If a witness is not found the machine outputs ``Don't-Know''. (We also assume that there are no pairs of contradicting witnesses for both $x\in L$ and $x\not\in L$.)

In \Cref{sec:apps} we show the following:

\begin{theorem*}[equivalence for average-case time-bounded Kolmogorov Complexity; \Cref{thm:ac}]
$\kpolys{n-O(1)}$ is zero-error average-case hard against $\NP/\poly$ machines iff \mbox{$\kpolys{n-o(n)}$} is zero-error average-case hard against $\NP/\poly$ machines.
\end{theorem*}

%
\subsubsection{Applications in proof complexity} In proof complexity, Kraj\'{i}\v{c}ek \cite{Kra04,Kra10-forcing} and Alekhnovich, Ben-Sasson, Razborov and Wigderson~\cite{ABSRW00} developed the theory of \emph{proof complexity generators}. Given a \ppoly\ mapping $g:\{0,1\}^n\to\{0,1\}^\ell$, with $n<\ell$, and a fixed vector $r\in\bits^\ell$, we denote by $\tau(g)_r$ the $\poly(\ell)$-size propositional formula that encodes naturally the statement  $r\not\in{\textrm {Im}}(g)$, so that if $r$ is  not in the image of $g$ then $\tau(g)_r$ is a propositional tautology. For $r$ not in the image of $g$, the tautology $\tau(g)_r$ is called a \emph{proof complexity generator}, and the hope is that for strong propositional proof systems one can establish (at least conditionally) that there are no $\poly(\ell)$-size proofs of $\tau(g)_r$, under the assumption that the mapping $g$ is sufficiently pseudorandom (see also~\cite{Razb15-annals}). Kraj\'{i}\v{c}ek  observed the connection between proof complexity generators and demi-bits (see \cite[Corollary 1.3]{Kra04} and the discussion that follows there.)

An immediate corollary of \autoref{thm:alg} is the following.

\begin{corollary*}[Stretching proof complexity generators; \Cref{cor:proof-complexity-generators}]
Let $b:\{0,1\}^n\to\{0,1\}^{n+1}$ be a demi-bit computable in \ppoly. Let $0<c<1$ be a constant and $\ell=n+n^c$. Then, there is a proof complexity generator  $g:\{0,1\}^n\to\{0,1\}^{\ell}$  in \ppoly, such that for every  propositional proof system, with probability at least  $1-\frac{1}{\ell^{\omega(1)}}$ over the choice of $r\in\bits^{\ell}$, there are no $\poly(n)$-size proofs of the tautology $\tau(g)_r$.\footnote{The points $r$ are taken uniformly from $\bits^\ell$, and with probability $1-1/2^{\ell-n}$ the formula $\tau(g)_r$ is a tautology, because for all  $r\in \bits^\ell\setminus {\textup{Im}}(g)$ the formula $\tau(g)_r$ is a tautology. While in some works, proof complexity generators are supposed to be hard for \emph{every} $r$ outside the image of the generator $g$, in our formulation the hardness is only with high probability over the $r$'s. It is unclear whether this makes any difference in the theory of proof complexity generators, since we are not aware of a case where the property that $\tau(g)_r$ is hard for every $r\not\in\textup{Im}(g)$ is used (although all cases of provably hard proof complexity generators against weak proof systems we  know of, are hard for every $r\not\in\textup{Im}(g)$).}
\end{corollary*}



\subsubsection{Technique overview}
We prove the stretchability of demi-bits by a novel and more flexible use of the hybrid argument combined with  other ideas.

The \textit{hybrid argument} (a.k.a.~the hybrid method, the hybrid technique, etc.) is a common proof technique  originating from the work of Goldwasser and Micali \cite{hybrid}. It was named by Leonid Levin.
(See \Cref{sec:notations} for notations used below.)
When we have a generator $g : \bi{n} \to \bi{m(n)}$, a distinguisher $D$, and a function $p$ (usually a polynomial) such that
\begin{equation*}
\pr{D(U_m) = 1} - \pr{D(g(U_n)) = 1} \geq 1/p(n), 
\tag{$\ast$}
\end{equation*}
where $U_m$ stands for the truly random strings and $g(U_n)$ stands for the pseudorandom ones,
the standard hybrid argument defines a spectrum (i.e. an ordered set) of random variables $H_i$'s, called hybrids, traversing from one extreme, $U_m$, to another, $g(U_n)$. A concrete example is $H_i := g(U_n)[1...i] \cdot U_{m - i}, 0 \leq i \leq m$ (where  $\cdot$ here means concatenation).  In this example, indeed $H_0 = U_m$ and $H_m = g(U_n)$.
Then the inequality $(\ast)$ can be written as:
\begin{align*}
    1/p(n)
    &\leq
    \pr{D(U_m) = 1} - \pr{D(g(U_n)) = 1}
    \\&=
    \sum_i (\pr{D(H_i) = 1} - \pr{D(H_{i+1}) = 1}).
\end{align*}
Thus, a usual next step is to claim there exists some $i$ such that
\[
\pr{D(H_i) = 1} - \pr{D(H_{i+1}) = 1}
\geq
\frac{1}{k \cdot p(n)},
\]
where $k$ is the total number of hybrids (in the above example, $k = m$).
In a nutshell, we show that if we can distinguish $U_m$ from $g(U_n)$ by a $1/p(n)$ portion, then we can distinguish some neighbouring pair of hybrids $H_i$ from $H_{i+1}$ by a $1/(k \cdot p(n))$ portion. See \Cref{lemma:learn} for a simple demonstration of the hybrid argument.

As  mentioned above, a standard hybrid argument cannot be applied to prove a stretched generator $g$ by some stretching algorithm, from a single demi-bit $b$, is still demi-bits. We now intuitively explain the reason for this.
A usual proof goes like this: 
we assume, for a contradiction, $g$ are not demi-bits. Then there are some distinguisher $D$ of $g$ and a function $p$ such that 
$\pr{D(U_m) = 1} - \pr{D(g(U_n)) = 1} \geq 1/p(n)$, 
and in particular $\pr{D(g(U_n)) = 1} = 0$ as $D$  breaks  demi-bits, and we hope to  construct a new appropriate  distinguisher $C$ of $b$ based on $D$.
However, as we saw above, a standard hybrid argument only yields that
$\pr{D(H_i) = 1} - \pr{D(H_{i+1}) = 1} \geq 1/p'(n)$ for some function $p'$ and cannot deduce that $\pr{D(H_{i+1}) = 1} = 0$. Hence, it is unclear how to continue this construction. 

The argument for proving  \autoref{thm:alg} proceeds by the contrapositive. We assume there is a distinguisher $D$ which breaks demi-bits $g$ (stretched from a single demi-bit $b$ by \Cref{alg:stretch}) in the desired sense, and we want to construct a distinguisher $C$ which breaks $b$.
Rather than applying the hybrid argument directly to $D$, 
we apply the hybrid argument to a new distinguisher $D'$ defined based on $D$: 
the new distinguisher $D'$ can use nondeterminism to change the pseudorandom part of the hybrids and thus ``amplifies'' the probability of certificating randomness (intuitively, this can  be viewed as changing the average-case analysis in the standard hybrid argument to a worst-case or existence analysis). By applying the hybrid argument to $D'$, we are able to identify a non-empty class, denoted by $S_2$ in the proof, of random strings $y_{i+1} \ldots y_m$, that are not random witnesses (in the sense that, for each $y_{i+1} \ldots y_m$ in this class, there are no seeds $x_1, \ldots, x_i$ such that $D(b(x_1) \ldots b(x_i) \: y_{i+1} \ldots y_m) = 1$).
Thus, for $y_{i+1} \ldots y_m$ in $S_2$, 
$y_{i} y_{i+1} \ldots y_m$ can become a random witness (i.e., there are seeds $x_1, \ldots, x_{i-1}$ such that $D(b(x_1) \ldots b(x_{i-1}) \: y_{i} \ldots y_m) = 1$) only if $y_{i}$ is truly random (i.e., not equal to $b(x)$ for some seed $x$). 
The hybrid argument also implies a ``good" such $y_{i+1} \ldots y_m$ in $S_2$, which can identify a sufficient portion of truly random $y_{i}$. We can thereby build a new distinguisher $C$ to distinguish truly random strings from pseudorandom ones.

A key step that makes this proof work is that nondeterministically guessing seeds $x_1, \ldots, x_{i-1}$ in $b(x_1) \ldots b(x_{i-1}) z_{i}$ preserves the ``randomness-structure'' of $b(x_1) \ldots b(x_{i-1}) z_{i}$, in the sense that: when $z_{i} = b(\cdot)$ is pseudorandom, the nondeterministic guess preserves the form $b(\cdot) \ldots b(\cdot) b(\cdot)$ (i.e., $i$ equal-length pseudorandom chunks); and when $z_{i} = y$ is truly random, it preserves the form $b(\cdot) \ldots b(\cdot) y$ (i.e., $i-1$ equal-length pseudorandom chunks followed by a truly random chunk $y$ of the same length). 
For common stretching algorithms that produce exponentially many new bits (e.g., recursively applying a one-bit generator), it is unclear how to use nondeterminism in a way that respects the ``randomness-structure'' of a given string.
Nevertheless, the new proof technique should hopefully inspire researchers to further explore the stretchability of demi-bits. 
On the other hand, the fact could also be that there is a specific demi-bit which cannot be stretched to exponentially many demi-bits by the standard stretching algorithms which are applied to super-bits and strong PRGs.

\subsection{Fine-grained characterisation of nondeterministic security based on unpredictability
}\label{sec:overv:big-formula}

Yao \cite{Yao1} defined PRGs as producing sequences that are computationally indistinguishable, by deterministic adversaries, from uniform sequences and proved that this definition of indistinguishability is equivalent to deterministic unpredictability, which was used in an earlier definition of PRGs suggested by Blum and Micali \cite{Blum_Micali}. Loosely speaking, unpredictability means, given any strict prefix of a random string, it is infeasible to predict the next bit.

In \autoref{sec:predict}, we provide a more fine-grained picture of nondeterministic hardness  (\Cref{def:nondeterministic-hardness}), by introducing the concept of nondeterministic unpredictability. This allows us to establish new lower and upper bounds to nondeterministic hardness, in the sense that we sandwich nondeterministic hardness between two unpredictability properties.



Specifically, we propose four  notions of unpredictability for probability ensembles:
\begin{enumerate}
    \item \uline{$\NP/\poly$-unpredictability}: 
    the capacity of being unpredictable by $\NP/\poly$ predictors.
    
    \item\uline{ $\coNP/\poly$-unpredictability}:
    the capacity of being unpredictable by $\coNP/\poly$ predictors.
    
    \item \uline{$\cup$-unpredictability}:
    the capacity of being unpredictable by  predictors in the union of $\NP/\poly$ and $\coNP/\poly$.
    
    \item\uline{$\cap$-unpredictability}:
    the capacity of being unpredictable by nondeterministic function-computing predictors.
\end{enumerate}

The names $\NP/\poly$-unpredictability,  $\coNP/\poly$-unpredictability, and $\cup$-unpredictability (a shorthand for $\NP/\poly \cup \coNP/\poly$-unpredictability) are self-explanatory, while the use of ``$\cap$-unpredictability" is somewhat less intuitive. We will see, in \Cref{sec:nd-pred} (where we use nondeterministic function-computing machines\footnote{We say a nondeterministic algorithm $\McA$ is a \textbf{function-computing} algorithm, 
if for every input $x \in \bi{n}$,
every computation branch yields one of $\{0, 1, \bot \}$, in which $\bot$ indicates a failure,
and there is always a computation branch yielding $0$ or $1$.}), that a decision problem is in $\NP/\poly \cap \coNP/\poly$ if and only if it is decidable by a nondeterministic polynomial-size function-computing algorithm. 

We establish the following characterisation of the nondeterministic hardness of generators from an unpredictability perspective:

\begin{theorem*}[\Cref{sum:unpredict}]
Here, $A \leq B$ means that if a generator has property $B$, then it also has property $A$:
%
%
\begin{figure}[H] 
\psscalebox{1.0 1.0} 
{
\begin{pspicture}(0,-3.9357579)(16.473505,-0.4433162)
\definecolor{colour1}{rgb}{0.9019608,0.9019608,0.9019608}
\psframe[linecolor=black, linewidth=0.04, fillstyle=solid,fillcolor=colour1, dimen=outer, framearc=0.35](12.344411,-0.4433162)(7.6175084,-1.51537)
\psrotate(7.2288127, -1.196827){26.094822}{\rput[bl](7.2288127,-1.196827){$\le$}}
\rput[bl](0.7648719,-1.6459315){\ppoly-}
\psrotate(7.0910635, -2.3582706){-22.680593}{\rput[bl](7.0910635,-2.3582706){$\le$}}
\rput[bl](0.176633,-1.9793179){\textup{-unpredictability}}
\rput[bl](3.1704955,-1.8564247){$\le~ \cap$\textup{-unpredictability}}
\rput[bl](8.3399515,-0.8728388){super-polynomial}
\rput[bl](7.790072,-1.2687984){nondeterministic hardness}
\rput[bl](7.8272605,-2.5348604){  $\NP/\poly\textup{-unpredictability}$ }
\psrotate(12.681488, -2.5171003){30.20006}{\rput[bl](12.681488,-2.5171003){$\le$}}
\psrotate(12.568805, -1.1049927){-22.680593}{\rput[bl](12.568805,-1.1049927){$\le$}}
\rput[bl](13.080692,-1.7751957){$\cup$\textup{-unpredictability}}
\psframe[linecolor=black, linewidth=0.04, dimen=outer, framearc=0.35](3.1338384,-1.2075249)(0.0,-2.1895452)
\psframe[linecolor=black, linewidth=0.04, dimen=outer, framearc=0.35](12.50085,-1.9913346)(7.5969105,-2.7026687)
\psrotate(14.649411, -1.6727943){0.15163828}{\psframe[linecolor=black, linewidth=0.04, dimen=outer, framearc=0.35](16.34798,-1.2082994)(12.950842,-2.1372893)}
\psframe[linecolor=black, linewidth=0.04, dimen=outer, framearc=0.35](7.057273,-1.2351344)(3.5562289,-2.2171545)
\rput[bl](7.854062,-3.5563045){$\coNP/\poly\textup{-unpredictability}$ 
}
\psrotate(7.0743403, -2.9628682){-36.68}{\rput[bl](7.0743403,-2.9628682){$\le$}}
\psrotate(12.682867, -3.1422346){-316.8}{\rput[bl](12.682867,-3.1422346){$\le$}}
\psframe[linecolor=black, linewidth=0.04, dimen=outer, framearc=0.35](12.47403,-3.0434418)(7.5700903,-3.7547758)
\end{pspicture}
}
%
\caption{Super-polynomial nondeterministic hardness here refers to \cref{def:nondeterministic-hardness}. Note that $\cup$\textup{-unpredictability} is at least as strong as $\NP/\poly\textup{-unpredictability}$, because it rules out predictors in \emph{both} $\NP/\poly$ and $\coNP/\poly$. And similarly, $\cup$\textup{-unpredictability} is at least as strong as $\coNP/\poly$\textup{-unpredictability.}}
\end{figure}
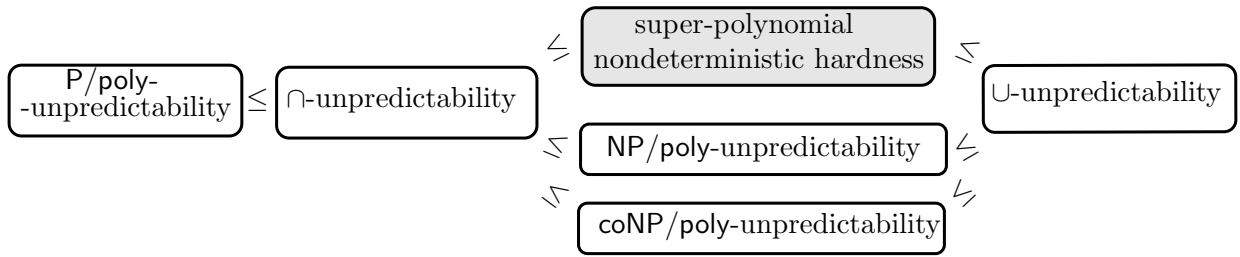
\end{theorem*}
\subsection{Super-cores:  hard-core predicates in the nondeterministic setting} 


In the deterministic context, the existence of strong PRGs is known to be equivalent to the existence of central cryptographic primitives such as one-way functions, secure private-key encryption schemes, digital signatures, etc. Liu and Pass  \cite{Liu2020} recently showed that a meta-complexity assumption  about mild average-case hardness of the time-bounded Kolmogorov Complexity is also equivalent to the existence of strong PRGs.
On the other hand, in the nondeterministic case we  have no known  such equivalent characterisations (of  nondeterministic-secure PRGs, namely, super-bits). In \autoref{sec:hard-core}, we introduce a definition of \emph{super-cores} serving as a nondeterministic variant of hard-cores. We then use this concept to draw the first equivalent characterisation of super-bits.

%
%

We start by reviewing the concepts of one-way functions and hard-core predicates. Loosely speaking, a one-way function (family) $f$ is a one that is easy to compute but hard to invert on average (with the probability taken over the domain of $f$). More precisely, ``easy to compute" means $f$ is in $\P$ or \ppoly{}, and ``hard to invert" means any efficient deterministic algorithm can only invert a negligible portion of $y = f(x)$ when $x$ is unseen. By ``efficient", in the uniform setting, we mean an algorithm in bounded-error probabilistic polynomial time ($BPP$), and in the nonuniform setting, an algorithm in \ppoly{}, and by ``invert", we mean finding an $x'$ for a given $y$ in $range(f)$ such that $f(x') = y$. A negligible portion for us means a portion that is less than $1/p(n)$ for any polynomial $p$ and all large $n$'s. We say $b: \bi{n} \to \bi{}$ in $\P$ or \ppoly{} is \emph{a hard-core of a function} $f$ if it is impossible to efficiently predict $b(x)$ with probability at least $1/2 + 1/\poly(n)$  given $f(x)$.

The existence of hard-core predicates is known (e.g., $b(x) = x[-1]$, the last bit of a string $x$, is a hard-core of the function $f(x) = x[1]$, the first bit of $x$), but the existence of a hard-core for a one-way function and the existence of any one-way function to begin with are unknown.
Goldreich and Levin \cite{hard_core_from_OWF} proved that inner product $mod\ 2$ is a hard-core for any function of the form $g(x,r) = (f(x),r)$, where $f$ is any one-way function and $|x| = |r|$. Subsequently, H\r{a}stad, Impagliazzo, Levin, and Luby \cite{PRG_from_OWF} showed that strong PRGs exist if and only if one-way functions exist. This theorem can  be stated equivalently as: a strong PRG exists if and only if a hard-core of some one-way function exists.

Since a strong PRG exists if and only if  a hard-core of some one-way function exists, and a super-bit is the nondeterministic analogue  of strong PRGs, a meaningful question to ask is: \smallskip
\textit{\begin{adjustwidth}{1.2cm}{1.2cm}
    What are the nondeterministic analogues   of one-way functions and hard-cores? 
\end{adjustwidth}}
\smallskip
To come up with a reasonable definition of super one-way functions is not an easy task because, for any function $f$, a nondeterministic algorithm can always invert a range-element $y$ by guessing some $x$ and checking if $f(x) = y$. Similarly, a reasonable definition of super-core predicates is non-trivial as well: for any function $f$ and predicate $b$, a nondeterministic algorithm can predict $b(x)$ when given $f(x)$ as input by guessing $x$ and then applying $b$.

We propose a definition of super-cores of a function $f$, which are secure against both \nppoly{} and \coNP/\poly\ predictors in the sense of \Cref{def:super-core}, when $f(x)$ is presented as the input.
With this definition, we can establish the following equivalence (we say a function $f: \bi{n} \to 
\bi{m(n)}$ is \emph{non-shrinking} if $m(n) \geq n$ for every $n$):


\begin{theorem*}[Informal; \autoref{equiv4}]
There is a super-core of some non-shrinking function if and only if there is a super-bit.
\end{theorem*}
This result is analogous to the  known equivalence between the existence of a hard-core of some one-way function and the existence of a strong PRG.
We also show that  a certain class of functions, which may have hard-cores, cannot possess any super-core. This also suggests that a one-way function could possibly possess no super-cores.
(See the text before \Cref{1-1} for the definition of ``predominantly one-to-one''.)

\begin{theorem*}[Informal; \autoref{1-1}]
If a  non-shrinking  function $f$ is ``predominantly'' one-to-one, then $f$ does not have a super-core.
\end{theorem*}
%

%
What we achieve in \autoref{sec:hard-core} provides a  step forward to better understand the nondeterministic hardness of PRGs and develop a sensible definition of one-way functions in the nondeterministic setting.

\subsection{Justification for the existence of demi-bits and PAC-learning}

Here we mention the work of Pich~\cite{Pich} who demonstrated the plausibility of the existence of demi-bits, by showing that if the class of polynomial-size boolean circuits is not PAC-learnable by sub-exponential circuits then demi-bits exist (where the learner is allowed to generate a nondeterministic or co-nondeterministic algorithm approximating the target function). Thus, if the class of polynomial-size boolean circuits is not PAC-learnable by sub-exponential circuits and the existence of demi-bits implies the existence of super-bits, there are no \tildenppoly-natural properties useful against \ppoly{}.
We provide an exposition of these results for self-containedness.  

Recall that  \textbf{PAC-learning} algorithms can be developed from breaking PRGs (e.g., cf. \cite{crypto, learn_Oliveira_Santhanam}).
The model of PAC-learning (an abbreviation of probably approximately correct learning) was introduced and developed by Valiant in \cite{pac1, pac2, pac3}.
In the PAC-learning model, the goal of a learner is to learn an arbitrary target function $f$ drawn from a target set (e.g., the class of decision trees, the class of boolean conjunctions, \ppoly{}, etc.). 
The target function is invisible to the learner.
The learner receives samples (randomly or by querying) from an $f$-oracle and selects a generalization function $f'$, called the hypothesis, from some hypothesis class.
The selected function must have low generalization error (the ``approximately correct" in ``PAC") with high probability (the ``probably" in ``PAC"). %
Furthermore, the learner is expected to be efficient and to output hypotheses that can as well be efficiently evaluated on any given input.

%
%
%
%
%
%
%
%
%
%
%
%
%
\begin{theorem*}[Informal; \autoref{thm:learn} Pich \cite{Pich}]
Demi-bits exist assuming the class of polynomial-size boolean circuits is not PAC-learnable by sub-exponential circuits, where the learner is allowed to generate a nondeterministic or co-nondeterministic circuit approximating the target function.
\end{theorem*}

This result demonstrates the plausibility of the existence of demi-bits because it is widely believed (e.g., cf. \cite{learn_Rajgopal_Santhanam}) that learning \ppoly{} is hard. And accordingly, it is reasonable to assume that if the learner is allowed to generate a nondeterministic or co-nondeterministic circuit family, the task remains hard. 

In contrast to demi-bits, it is worth mentioning that it is unknown how to relate the existence of super-bits to the hardness of PAC-learnable (namely, the proof of \autoref{thm:learn} does not carry through when we consider super-bits).


\bigskip



\section{Preliminaries and basic concepts} \label{sec2}

\subsection{Notations and conventions} \label{sec:notations}

We follow the following conventions:
\begin{itemize}
\item $\bN$ denotes the set of positive integers (excluding 0). For $n \in \bN$, $[n]$ denotes $\{1, ..., n\}$. $[0]$ is the empty set $\emptyset$.

\item The size of a Boolean circuit C, denoted as $size(C)$ or $|C|$, is the total number of gates (including the input gates).
$Circuit[s]$ denotes the Boolean circuits of size at most s.
If $s: \bN \to \bN$ is a function, $Circuit[s]$ contains all the Boolean circuit families $C_n$ such that $|C_n| \leq s(n)$ for all large $n$'s.

\item All the distributions we consider in this work are, by default, uniform. $U_n$ denotes the uniform distribution over $\bi{n}$ unless stated  otherwise.

\item For functions $f, g: \{0, 1\}^n \to \{0, 1\}$, we say $g$ $\gamma$-approximates $f$ if $\bP_x[f(x) = g(x)] \geq \gamma$.

\item For a string $x$, where its bits are indexed from left to right by $1$ to $|x|$, $x[i]$ denotes its $i$-th bit, and $x[i...j]$ denotes the sub-string indexed from $i$ to $j$ of $x$ (if $i > j$, $x[i...j] = \epsilon$, the empty string). $x[-i]$ denotes its $i$-th last bit.

\item For strings $x, y$, we may use any of the following to denote the concatenation of $x$ and $y$: $xy$, $(x,y)$, $x \cdot y$.

\item We may not verbally distinguish a function with its string representation (this can be a truth table, or a string encoding a circuit representation of this function, etc.) when there is no ambiguity.

\end{itemize}
We may also follow other common conventions used in the complexity community or literature.

\subsection{Computational models}

\begin{definition}[Randomized  circuits; equiv.~probabilistic circuit]
A circuit $C$ is a \textbf{randomized circuit} (equivalently, a \emph{probabilistic circuit}) if, in addition to the standard input bits (similar to the input bits of a non-randomized circuit), it contains zero or more \emph{random input bits} (i.e., bits taken from a random distribution).
We call $\{C_n\}_{n=1}^\infty$ a \textbf{randomized circuit family} if for every $n$, $C_n$ is a randomized circuit with $n$ standard input bits.
\end{definition}

Note that if a randomized circuit family $\{C_n\}$ is in $Circuit[s(n)]$, it  means that for every sufficiently large $n$, the randomized circuit $C_n$ has size at most $s(n)$, which automatically constrains the number of  random input bits that $C_n$ is allowed to have.

\begin{definition}[Nondeterministic and co-nondeterministic circuits]\label{def:nondeterministic-circuit}
A circuit $C(\overline x,\overline r)$ is a \mbox{\textbf{(co-)nondeterministic circuit}} if, in addition to the standard input bits $\overline x$, it contains zero or more \emph{nondeterministic input bits} $\overline r$ (namely, bits that are meant to control the nondeterministic decisions made by the circuit). A nondeterministic circuit with a single output bit is said to \textbf{accept} an input $\overline \alpha\in\bits^n$ to $\overline x$ iff there \emph{exists} an  assignment $\overline \beta \in \bits^{|\overline r|}$ to $\overline r$ such that $C(\overline \alpha,\overline \beta)=1$ (and otherwise it is said to reject  $\overline x$). A co-nondeterministic circuit with a single output bit is said to \textbf{reject} an input $\overline \alpha\in\bits^n$ to $\overline x$ iff there exists an assignment $\overline \beta \in \bits^{|\overline r|}$ to $\overline r$ such that $C(\overline \alpha,\overline \beta)=0$ (and otherwise it is said to accept $\overline x$). We call $\{C_n\}$ a (co-)nondeterministic circuit family if for every $n$, $C_n$ is a (co-)nondeterministic circuit with $n$ standard input bits.

\end{definition}

\begin{definition}[Oracle circuits]
$C$ is an \textbf{oracle circuit} if it is allowed to use  oracle gates. We write $C$ as $C^{f_1,..., f_k}$ if $C$ has oracle gates computing Boolean functions $f_1,..., f_k$.
\end{definition}

We note that an oracle gate computing a Boolean function $f: \bi{n} \to \bi{}$ has fan-in $n$, and in our model, an oracle gate is allowed to appear in any place in  the circuit.
\subsection{Natural proofs}
Let $F_n$ be the set of all functions $f: \bi{n} \to \bi{}$ and $\Gamma \text{ and } \Lambda$ be complexity classes. We call $C = (C_n)_{n \in \bN}$ a \textbf{combinatorial property} of boolean functions if each $C_n \subseteq F_n$.
\begin{definition}[Natural properties \cite{natural}] \label{def: natural}
We say a combinatorial property $C = (C_n)_{n \in \bN}$ is $\boldsymbol{\Gamma}$\textbf{-natural} if some $C' = (C_n')_{n \in \bN}$ with $C_n' \subseteq C_n$ for each $n$ satisfies:
\begin{itemize}
    \item \textbf{Constructivity.}
    Whether $f \in C_n'$ is computable in $\Gamma$ when $f$ is encoded by its truth table as input.
    \item \textbf{Largeness.}
    $|C_n'| \geq 2^{-O(n)} \cdot |F_n|$ for all large $n$'s.
\end{itemize}
We say $C$ is \textbf{useful} against $\boldsymbol{\Lambda}$ if it satisfies:
\begin{itemize}
    \item \textbf{Usefulness.} For any function family $f = (f_n)_{n \in \bN}$, if $f_n \in C_n$ infinitely often, then $f \notin \Lambda$.
\end{itemize}
\end{definition}
A circuit lower bound proof (that some function family is not in $\Lambda$) is called a $\Gamma$-\textbf{natural proof} against $\Lambda$ if it uses, explicitly or implicitly, some $\Gamma$-natural combinatorial property useful against $\Lambda$. Especially, a \ppoly-natural proof against \ppoly{} is a proof that uses a \ppoly-natural combinatorial properties useful against \ppoly{}.

We note that the notion of natural proofs, unlike natural combinatorial property, is not defined in a mathematically rigorous sense. Nevertheless, the use of the terminology ``natural proof" in a statement more intuitively embodies our intention and also does not affect the rigorousness of the statement: whenever we say $\Gamma$-natural proofs against $\Lambda$ do or do not exist, what we mean, in a mathematically rigorous sense, is $\Gamma$-natural combinatorial properties against $\Lambda$ do or do not exist.

\subsection{Pseudorandom generators}

We recall here the basic definition of pseudorandom generators. 
As mentioned in the introduction, all the distributions we consider in this work are, by default, uniform, and $U_n$ denotes the
uniform distribution over $\{0, 1\}^n$ unless stated otherwise.

\begin{definition}[Generators]
A function family $g_n: \{0, 1\}^n \to \{0, 1\}^{l(n)}$ is a \textbf{generator} if $g_n \in \P/\poly$ and $l(n) > n$ for every $n$. We call such an $l$ a \emph{stretching function} and call $l(n) - n$ the \emph{stretching length} of $g_n$ (sometimes  $l(n)$ is called \emph{the stretching length}).
\end{definition}
We note that all the generators in this work will be computable  in \ppoly{}, although
in more general settings, it is not required that generators are \ppoly-computable (cf. \cite{hardness_vs_randomness}). 

\begin{definition}[Standard hardness]\label{def:standard-hardness}
Let $g_n: \{0, 1\}^n \to \{0, 1\}^{l(n)}$ be a generator. Then the \textbf{hardness} $H(g_n)$ of $g_n$ is the minimal $s$ for which there exists a (deterministic) circuit $D$ of size at most $s$ such that
\begin{displaymath}
\bigg|
\underset{y \in \{0, 1\}^{l(n)}}{\mathbb{P}} [D(y) = 1] - \underset{x \in \{0, 1\}^n}{\mathbb{P}} [D(g_n(x)) = 1]
\bigg|
\geq 
1/s(n).
\end{displaymath}
\end{definition}

The order of the two terms in the absolute value and the absolute value itself are immaterial since in the deterministic setting, we can always flip the output bit of a distinguisher $D$.

\begin{definition}[(Strong) pseudorandom generators (PRG)] \label{def:strong PRG}
A generator $g : \bi{n} \to \bi{l(n)}$ is called a (strong) \textbf{PRG} if for every $D$ in \ppoly, every polynomial $p$, and all sufficiently large $n$'s,
\begin{displaymath}
\bigg|
\underset{y \in \{0, 1\}^{l(n)}}{\bP} [D(y) = 1] 
- 
\underset{x \in \{0, 1\}^n}{\bP} [D(g(x)) = 1]
\bigg| 
<
1/p(n).
\end{displaymath}
\end{definition}

In other words, a strong PRG is defined to be a generator safe against all polynomial-size distinguishers. An alternative definition used in some texts is: a generator with hardness at least $2^{n^\epsilon}$ for some $\epsilon > 0$ and all large $n$'s, which defines a stronger PRG.

We shall say that a function $f(n) : \bN \to \bR^+$ is \emph{(at least) exponential }if there exists some $\epsilon > 0$ such that $f(n) \geq 2^{n^\epsilon}$ for all sufficiently  large $n$'s. A function is \emph{not} (at least) exponential if for every $\epsilon > 0$, there exist infinitely many $n$'s such that $f(n) < 2^{n^\epsilon}$; this latter condition  is equivalent to: there is an infinite monotone sequence $(n_i) \subseteq \bN$ such that $f(n_i)$ is sub-exponential in $n_i$ (i.e., $f(n_i) = 2^{n_i^{o(1)}}$).

The existence of a strong PRG is considered quite plausible because many intractable problems (e.g., factoring) seem to provide a basis for constructing  such generators (cf.~\cite{Gol-crypto-bookI}).
\begin{conjecture} Strong PRGs exist.
\end{conjecture}

Razborov and Rudich showed that the existence of a strong PRG rules out the existence of \ppoly-natural proofs useful against \ppoly{} \cite{natural}.

\medskip

Concrete examples of strong PRGs are unknown, as the existence of such a PRG implies $\P \neq \NP$ in the uniform setting and $\ppoly \neq \NP/\smallpoly$ in the nonuniform setting. Nevertheless, generators that can fool classes of weaker distinguishers were constructed (e.g., Nisan and Wigderson \cite{hardness_vs_randomness}).\footnote{For weak models, both complexity-theoretic generators and cryptographic generators are known. Complexity-theoretic generators fooling $AC^0$ was shown by  Nisan (which is the Nisan-Wigderson generator with PARITY as the hard function, and is earlier than \cite{hardness_vs_randomness}). Cryptographic generators are constructed for example in \cite{cryptoeprint:2016/580}. For \P/\poly, both complexity-theoretic generators and cryptographic generators are unknown. However, the assumptions needed for complexity-theoretic generators (e.g, \textsf{E} requires exponential-size) are much weaker than those needed for cryptographic generators (e.g., that one way functions exists).}


\subsection{Super-bits and demi-bits} \label{sec:sup-demi-pre}

Here we  provide a brief  review of the main results and open problems in \cite{superbit} that are relevant to our work.
(Some of the text is repeated from the introduction.) 
\begin{definition}[Nondeterministic hardness]
Let $g_n: \{0, 1\}^n \to \{0, 1\}^{l(n)}$ be a generator. Then the \textbf{nondeterministic hardness} $H_{\textup{nh}}(g_n)$  (also called \textbf{super-hardness}) of  $g_n$ is the minimal $s$ for which there exists a nondeterministic circuit $D$ of size at most $s$ such that
\begin{equation*}
    \underset{y \in \{0, 1\}^{l(n)}}{\mathbb{P}} [D(y) = 1] - \underset{x \in \{0, 1\}^n}{\mathbb{P}} [D(g_n(x)) = 1] \geq \frac{1}{s}.
    \tag{1}
\end{equation*}

\end{definition}
In contrast to the definition of deterministic hardness, the order of the two possibilities on the left-hand side is crucial. This order forces a nondeterministic distinguisher to certify the randomness of a given input. Reversing the order or keeping the absolute value trivialize the task of breaking $g$: a distinguisher $D$ can simply guess a seed $x$ and check if $g(x)$ equals the given input. For such a $D$, we have
$\pr{D(g(x))=1} = 1$ and $\pr{D(y) = 1} \leq 1/2$.

We call exponentially super-hard generators super-bits:
\begin{definition}[Super-bits]
A generator (in \P/\poly)
$g_n: \{0, 1\}^n \to \{0, 1\}^{n + c}$ for some $c: \bN \to \bN$ is called $c$ \textbf{super-bit(s)} (or a $c$-super-bit(s)) if $H_{\textup{nh}}(g_n) \geq 2^{n^\epsilon}$ for some $\epsilon > 0$ and all sufficiently large $n$'s. In particular, if $c = 1$, we call $g_n$ a super-bit.
\end{definition}

The term \textit{super-bits} thus stands for pseudorandom bits that can fool ``super" powerful adversaries. Rudich constructed a candidate super-bit based on the subset sum problem and conjectured that:
\begin{conjecture}[Super-bit conjecture]
There exists a super-bit.
\end{conjecture}

The main  theorem in \cite{superbit} is the following one, which is proved based on the stretchability of super-bits as  discussed above.
\begin{theorem}[\cite{superbit}]\label{thm:Rudich}
If super-bits exist, then there are no $N\tilde{P}/qpoly$-natural properties useful against \ppoly{}, where $N\tilde{P}/qpoly$ is the class of languages recognised by non-uniform, quasi-polynomial-size circuit families (where quasi-polynomial means $n^{log^{O(1)}(n)}$).
\end{theorem}
We remark that, in this theorem, the ``largeness" requirement of $N\tilde{P}/qpoly$-natural properties can in fact be relaxed to $|C_n'| \geq 2^{-n^{O(1)}} \cdot |F_n|$ (cf. \Cref{def: natural}).

Rudich also proposed another notion, called demi-hardness,  which he considered to be more intuitive than super-hardness:
\begin{definition}[Demi-hardness]
Let $g_n: \{0, 1\}^n \to \{0, 1\}^{l(n)}$ be a generator (in \ppoly). Then the \textbf{demi-hardness} $H_{\textup{dh}}(g_n)$ of  $g_n$ is the minimal $s$ for which there exists a nondeterministic circuit $D$ of size at most $s$ such that
\begin{equation*}
    \underset{y \in \{0, 1\}^{l(n)}}{\mathbb{P}} [D(y) = 1] \geq \frac{1}{s}
    \text{\quad and \ }
    \underset{x \in \{0, 1\}^n}{\mathbb{P}} [D(g_n(x)) = 1] = 0.
    \tag{2}
\end{equation*}
\end{definition}
We note that (2), which requires a distinguisher to make no mistakes on any generated strings, is a stronger requirement than (1). Thus,
$H_{\textup{nh}}(g) 
\leq
H_{\textup{dh}}(g)$ for every generator $g$.

We call exponentially demi-hard generators demi-bits, where ``demi" is meant to stand for ``half'' here:
\begin{definition}[Demi-bits]
A generator (in \ppoly) $g_n: \{0, 1\}^n \to \{0, 1\}^{n + c}$ for some $c: \bN \to \bN$ is called $c$ \textbf{demi-bit(s)} (or a $c$-demi-bit(s)) if $H_{\textup{dh}}(g_n) \geq 2^{n^\epsilon}$ for some $\epsilon > 0$ and all sufficiently large $n$'s. In particular, if $c = 1$, we call $g_n$ a demi-bit.
\end{definition}

As $H_{\textup{nh}}(g) \leq H_{\textup{dh}}(g)$, it is natural to conjecture:
\begin{conjecture}[Demi-bit conjecture \cite{superbit}]
There exists a demi-bit.
\end{conjecture}
Accordingly, it is natural to  ask:
\begin{open problem*}[\cite{superbit}]
Does the existence of a demi-bit imply the existence of a super-bit?
\end{open problem*}

As  discussed above, a super-bit can be stretched to polynomially many super-bits and to pseudorandom function generators secure against nondeterministic adversaries. In contrast, whether a demi-bit can be stretched to even two demi-bits was unknown prior to our work:
\begin{open problem*}[\cite{superbit}; Resolved in \autoref{sec:strectching-demi-bits}]
Given a demi-bit, is it possible to stretch it to  2-demi-bits?
\end{open problem*}

The question that remains open is:
\begin{open problem*}[\cite{superbit}]
Given a demi-bit, is it possible to build a pseudorandom function generator with exponential demi-hardness?
\end{open problem*}

A positive answer to the last problem would answer the following:

\begin{open problem*}[\cite{superbit}]
Does the existence of a demi-bit rule out the existence of $N\tilde{P}/qpoly$-natural properties against \ppoly?
\end{open problem*}

\subsection{Infinitely often super-bits and demi-bits}

Here, we formalize a weaker variant of super-bits and demi-bits, that  only requires infinitely many $n$'s to be ``hard" (recall super-bits/demi-bits require hardness for all sufficiently large $n$'s). This variant occasionally appears implicitly in the literature but may not have been formally defined.

\begin{definition}[Infinitely often super-bits/demi-bits]
A generator (in \ppoly) $g_n: \{0, 1\}^n \to \{0, 1\}^{n + c}$, for some $c: \bN \to \bN$ is called $c$ \textbf{infinitely often} (\textbf{i.o.})~\textbf{super-bit(s)/demi-bit(s)}, if 
$H_{\textup{nh}}(g_n) \geq 2^{n^\epsilon}$/$H_{\textup{dh}}(g_n) \geq 2^{n^\epsilon}$ 
for some $\epsilon > 0$ and infinitely many $n$'s. In particular, if $c = 1$, we call $g_n$ an i.o.~super-bit/demi-bit.
\end{definition}

In fact, it is an easy observation that we can construct from a reasonably frequent i.o.~super-bit/demi-bit, a super-bit/demi-bit, by properly choosing a prefix of a given input $n$ and applying the i.o.~algorithm to the prefix. More details follow. However, we are unaware if we can construct a super-bit/demi-bit from any i.o.~super-bit/demi-bit.

\begin{lemma}\label{i.o.}
Assume $g_n: \{0, 1\}^n \to \{0, 1\}^{n + c}$ for some $c \in \bN$ are $c$ i.o. super-bits/demi-bits. If there exist a polynomial $p$ and an infinite monotone sequence $(n_i)_{i \in \bN} \subseteq \bN$ such that:
(1) $n_{i+1} \leq p(n_i)$, and
(2)
for some $\epsilon > 0$ and every $n \in (n_i)_{i \in \bN}$, $H_{\textup{nh}}(g_n) \geq 2^{n^\epsilon}$/$H_{\textup{dh}}(g_n) \geq 2^{n^\epsilon}$, then there exists a $c$-super-bits/$c$-demi-bits constructed from $g_n$.
\end{lemma}
\begin{proof}
We present the proof for constructing super-bits from i.o. super-bits, and the proof for constructing demi-bits from i.o. demi-bits is almost identical.

Assume $g$, $(n_i)$, $p$ are as given in the lemma statement. Denote $m = n+c$ and $m_i = n_i + c$. We construct a new generator $G$ as follows:
\begin{adjustwidth}{1.2cm}{1.2cm}
    Given $x_n \in \bi{n}$, there is an $i$ such that $n_i \leq n < n_{i+1}$. Define $G(x_n) = g(a) \cdot b$, where $a = x_n[1...n_i], b = x_n[n_i+1...n]$. We note $|G(x_n)| = |g(a)| + |b| = n + c$.
\end{adjustwidth}

We want to show that $G$ is indeed super-bits. Suppose, for a contradiction, $G$ is not. Then there exist an infinite monotone sequence $S \subseteq \bN$, a sub-exponential function $s$, and a distinguisher $D$ of size $s$ such that for every $n \in S$,
\begin{align*}
    1/s(n)
    \leq
    \pr{D(U_m)=1} - \pr{D(G(U_n))=1}
    =
    \pr{D(U_{m_i}U_{m-m_i})=1} - \pr{D(g(U_{n_i})U_{m-m_i})=1}
\end{align*}
Thus, for every $n \in S$, there exists a fixed string $w = w(n)$ such that
\[
1/s(n) 
\leq 
\pr{D(U_{m_i}w)=1} - \pr{D(g(U_{n_i})w)=1}.
\]

We now construct a new distinguisher $D'$ for $g$ as follows:
\begin{adjustwidth}{1.2cm}{1.2cm}
    Given $Y \in \bi{m_i}$ as input, if there exists an $n \in S$ such that $n_i \leq n < n_{i+1} (\leq p(n_i))$, $D'$ outputs $D(Yw(n))$, and otherwise $D'$ always outputs $0$ (which means $D'$ fails to do anything for such an $n_i$).
\end{adjustwidth}
Note that the (1) $n_{i+1} \leq p(n_i)$ assumption guarantees the efficiency of $D'$.
As $S$ is infinite and every $n \in S$ is between some $n_i$ and $n_{i+1}$, there are infinitely many $n_i$'s such that:
\begin{align*}
    \pr{D'(U_{m_i}) = 1} - \pr{D'(g(U_{n_i}))=1} 
    =
    \pr{D(U_{m_i}w)=1} -
    \pr{D(g(U_{n_i})w)=1}
    \geq
    1/s(n)
\end{align*}
This shows, for infinitely many $n_i$'s, $H_{\textup{nh}}(g(n_i)) \leq 1/s'(n)$ for some sub-exponential $s'$, which contradicts the assumption (2) in the lemma statement.
\end{proof}
\textbf{Remark.} \Cref{i.o.} is often used implicitly in the constructions of super-bits and demi-bits.

%
%

\section{Stretching demi-bits}\label{sec:strectching-demi-bits}

In this section, we answer affirmatively whether a demi-bit is stretchable, which had been an open problem from the original work of Rudich \cite{superbit}.


We propose \Cref{alg:stretch}, which stretches a single demi-bit to $n^c$ demi-bits for any $c<1$,
and verify the correctness of this  stretching algorithm (i.e., verify the exponential demi-hardness of the new elongated generator).
Intuitively, \Cref{alg:stretch} partitions a given seed into disjoint pieces and apply the $1$-demi-bit generator to each piece.
The algorithm proposed here is very similar to other stretching or amplification algorithms seen in the literature (e.g., cf. \cite{Yao1, Luby}).
The real difficulty of demi-bit stretching comes from to prove the correctness of the attempted stretch (i.e., to proof the stretching algorithm applied to preserve exponential demi-hardness).

\begin{stretchalgorithm}\label{alg:stretch}
Suppose $b_n:\bits^n\to\bits^{n+1}$ is a demi-bit and $0< c <1$ is a constant. We define a new generator $g:\bits^N\to\bits^{N+m}$ with input length $N$ and stretching length $m = \lceil N^c \rceil$
as follows: given input $x$ (of length $N$), let $n = \lfloor \frac{N}{m} \rfloor$ and $x = x_1 x_2 \ldots x_m r$, where each $x_i$ has length $n$, and define $g(x) = b(x_1)\ldots b(x_m)r$.
\end{stretchalgorithm}

The correctness proof proceeds by  contrapositive.
That is, we assume there is a distinguisher $D$ which breaks demi-bits $g$ (stretched from a single demi-bit $b$ by \Cref{alg:stretch}) in the desired sense, and we want to  construct a distinguisher $C$ which breaks $b$.
Rather than applying the hybrid argument directly to $D$, 
we apply the hybrid argument to a new distinguisher $D'$ defined based on $D$: 
the new distinguisher $D'$ can use nondeterminism to change the pseudorandom part of the hybrids and thus ``amplify'' the probability of certificating randomness (intuitively, this can  be viewed as changing the average-case analysis in the standard hybrid argument to a worst-case or ``existence'' analysis). By applying the hybrid argument to $D'$, we are able to identify a non-empty class $S_2$ of random strings $y_{i+1} \ldots y_m$, which are not random witnesses (in the sense that, for each $y_{i+1} \ldots y_m$ in this class, there are no seeds $x_1, \ldots, x_i$ such that $D(b(x_1) \ldots b(x_i) \: y_{i+1} \ldots y_m) = 1$).
Thus, for $y_{i+1} \ldots y_m$ in $S_2$, 
$y_{i} y_{i+1} \ldots y_m$ can become a random witness (i.e., there are seeds $x_1, \ldots, x_{i-1}$ such that $D(b(x_1) \ldots b(x_{i-1}) \: y_{i} \ldots y_m) = 1$) only if $y_{i}$ is truly random (i.e., not equal to $b(x)$ for some seed $x$). 
The hybrid argument also implies a ``good" such $y_{i+1} \ldots y_m$ in $S_2$, that  can identify a sufficient portion of truly random $y_{i}$. We can thereby build a new distinguisher $C$ to distinguish truly random strings from pseudorandom ones.

In the proof, we reserve the \textbf{bold face} for random variables.
\begin{theorem}(Main theorem for stretching demi-bits) \label{thm:alg}
The generator $g$ (with a sub-linear stretching-length), as defined in \Cref{alg:stretch}, has at least exponential demi-hardness.
\end{theorem}

\begin{proof}
Demi-bit $b$ and constant $c$ are as given in \Cref{alg:stretch}. Suppose, towards contradiction that $g$ does not have exponential demi-hardness. That is, there is a sub-exponential (in the input length, denoted by $N$) size nondeterministic circuit $D$ such that, for infinitely many $N$'s (recall $m = \lceil N^c \rceil$ and $n = \lfloor \frac{N}{m} \rfloor$, and without loss of generality, we may assume $m|N$), $\pr{D(\mathbf{y_1} \ldots \mathbf{y_m}) = 1} \geq 1/|D|$ and $\pr{D(b(\mathbf{x_1}) \ldots b(\mathbf{x_m})) = 1} = 0$, where $\mathbf{x_1}, \ldots, \mathbf{x_m}$ are totally independent length-$n$ random strings and $\mathbf{y_1}, \ldots, \mathbf{y_m}$ are totally independent length-$(n+1)$ random strings. Our aim is to efficiently break $b$ in the desired sense.

We define a new nondeterministic circuit $D'$ that  takes a pair of inputs: the first is the same input as $D$'s, that is, an $(N+m)$-bit string $y_1\ldots y_m$, and the second is $i \in \{0, 1, \ldots, m\}$ (with $i$ properly encoded):

\begin{quote}
   \textit{Given input $(y_1 \ldots y_m, i)$, $D'$ guesses $n$-bit strings $x_1, \ldots, x_i$ and does whatever $D$ does on $b(x_1) \ldots b(x_i) \: y_{i+1} \ldots y_m$.} 
\end{quote}

We observe that:
\begin{itemize}
    \item $\pr{D'(\mathbf{y_1}\ldots \mathbf{y_m}, 0) = 1} = \pr{D(\mathbf{y_1} \ldots \mathbf{y_m}) = 1} \geq 1/|D|$
    (by the definition of $D'$);
    \item
    $\pr{D'(b(\mathbf{x_1}) \ldots  b(\mathbf{x_m}), m) = 1} = \pr{D(b(\mathbf{x_1}) \ldots b(\mathbf{x_m})) = 1} = 0$ \\ (by assumption on $D$;  otherwise  $\pr{D'(b(\mathbf{x_1}) \ldots  b(\mathbf{x_m}), m) = 1} > 0$ would mean there exist $x_1, \dots, x_m$ such that 
    $D(b(x_1) \ldots b(x_m)) = 1$);
    \item
    $D'$ is also of sub-exponential size.
\end{itemize}

We can now apply the hybrid argument to $D'$:
\begin{align*}
    1/|D|
    &\leq
    \pr{D'(\mathbf{y_1}\ldots \mathbf{y_m}, 0) = 1}
    -
    \pr{D'(b(\mathbf{x_1}) \ldots  b(\mathbf{x_m}), m) = 1} \\
    &=
    \sum_i
    (
    \pr{D'(b(\mathbf{x_1}) \ldots b(\mathbf{x_{i-1}}) \: \mathbf{y_{i}} \ldots \mathbf{y_m}, i-1)}
    -
    \pr{D'(b(\mathbf{x_1}) \ldots b(\mathbf{x_i}) \: \mathbf{y_{i+1}} \ldots \mathbf{y_m}, i)}
    )
\end{align*}
yields there is an $i = i(N)$ (for infinitely many $N$'s) such that
\[
P_{i-1} - P_{i} \geq 1 / (m \cdot |D|),
\]
where 
$
P_{i-1} :=
\pr{D'(b(\mathbf{x_1}) \ldots b(\mathbf{x_{i-1}}) \mathbf{y_{i}} \mathbf{y_{i+1}} \ldots \mathbf{y_m}, i-1) = 1}
$
and \\
$
P_{i} :=
\pr{D'(b(\mathbf{x_1}) \ldots b(\mathbf{x_{i-1}}) b(\mathbf{x_i}) \mathbf{y_{i+1}} \ldots \mathbf{y_m}, i) = 1}$. As $m$ is sublinear in $n$, $m \cdot |D|$ is still sub-exponential in $N$.

We denote by $S = \bi{(n+1) \times (m - i)}$   the set of $(m-i)$-tuples of $(n+1)$-bit strings, and let 
\[
S_1 = 
\{
(y_{i+1}, \ldots, y_m) \in S: \exists (x_1, \ldots, x_i) \in \bi{n \times i} \: D(b(x_1) \ldots b(x_i) y_{i+1} \ldots y_m) = 1 
\},
\]
and 
\[
S_2 = S \setminus S_1.
\]

We note that:
\begin{itemize}
    \item $D'(b(x_1) \ldots b(x_{i-1})\: y_{i} y_{i+1} \ldots y_m, i-1) = 1$ if and only if $D'(O \: y_{i} y_{i+1} \ldots y_m, i-1) = 1$, where $O = 0^{(n+1)\cdot (i-1)}$ (because, by construction, $D'(\dots,i-1)$ ignores its first $i-1$ input strings);

    \item $D'(b(x_1) \ldots b(x_{i-1})b(x_{i}) \: y_{i+1} \ldots y_m, i) = 1$ if and only if $(y_{i+1}, \ldots, y_m) \in S_1$, and thus $\pr{\mathbf{y_{i+1}} \ldots \mathbf{y_m} \in S_1} = P_i$.
\end{itemize}

Therefore,
\begin{align*}
    P_{i-1}
    =\ &
    \pr{D'(
    b(\mathbf{x_1}) \ldots b(\mathbf{x_{i-1}})\mathbf{y_{i}}
    \mathbf{y_{i+1}} \ldots \mathbf{y_m},
    i-1) = 1} \\
    =\ &
    \pr{D'(
    b(\mathbf{x_1}) \ldots b(\mathbf{x_{i-1}})\mathbf{y_{i}}
    \mathbf{y_{i+1}} \ldots \mathbf{y_m},
    i-1) = 1 | \mathbf{y_{i+1}} \ldots \mathbf{y_m} \in S_1}\cdot 
    \pr{\mathbf{y_{i+1}} \ldots \mathbf{y_m} \in S_1} \ + \\
    &
    \pr{D'(
    b(\mathbf{x_1}) \ldots b(\mathbf{x_{i-1}})\mathbf{y_{i}}
    \mathbf{y_{i+1}} \ldots \mathbf{y_m},
    i-1) = 1 | \mathbf{y_{i+1}} \ldots \mathbf{y_m} \in S_2}\cdot 
    \pr{\mathbf{y_{i+1}} \ldots \mathbf{y_m} \in S_2} \\
    \leq\ &
    1 \cdot \pr{\mathbf{y_{i+1}} \ldots \mathbf{y_m} \in S_1} \ + \\
    &
    \pr{D'(
    O \: \mathbf{y_{i}}
    \mathbf{y_{i+1}} \ldots \mathbf{y_m},
    i-1) = 1 | \mathbf{y_{i+1}} \ldots \mathbf{y_m} \in S_2}\cdot 
    \pr{\mathbf{y_{i+1}} \ldots \mathbf{y_m} \in S_2} \\
    =\ &
    P_i +
    \pr{D'(
    O \: \mathbf{y_{i}}
    \mathbf{y_{i+1}} \ldots \mathbf{y_m},
    i-1) = 1 | \mathbf{y_{i+1}} \ldots \mathbf{y_m} \in S_2} \cdot 
    \pr{\mathbf{y_{i+1}} \ldots \mathbf{y_m} \in S_2},
\end{align*}
and thus
\[
\pr{D'(
    O \: \mathbf{y_{i}}
    \mathbf{y_{i+1}} \ldots \mathbf{y_m},
    i-1) = 1 | \mathbf{y_{i+1}} \ldots \mathbf{y_m} \in S_2} \cdot 
    \pr{\mathbf{y_{i+1}} \ldots \mathbf{y_m} \in S_2}
\geq
P_{i-1} - P_i
\geq
1 / (m \cdot |D|).
\]

In particular, 
$\pr{\mathbf{y_{i+1}} \ldots \mathbf{y_m} \in S_2} > 0$
\ and\,
$
\pr{D'(
    O \:
    \mathbf{y_{i}} \:
    \mathbf{y_{i+1}} \ldots \mathbf{y_m},
    i-1) = 1 |
    \mathbf{y_{i+1}} \ldots \mathbf{y_m} \in S_2}
    \geq 
    1 / (m \cdot |D|),
$
which imply there is a \emph{fixed} $(y_{i+1} \ldots y_m) \in S_2$ such that
\[
\pr{D'(
O \:
\mathbf{y_{i}} \:
y_{i+1} \ldots y_m,
i-1) = 1}
\geq 
1 / (m \cdot |D|).
\]

We now define another nondeterministic circuit $C$ by $C(\mathbf{y_{i}}) :=
D'(O \: \mathbf{y_{i}} y_{i+1} \ldots y_m, i-1)$ with $\mathbf{y_{i}} \in \bi{n+1}$ as the input variable.
As $D'$ is of sub-exponential size in $N$ and $n$ is polynomially related to $N$, $C$ is of sub-exponential size in $n$.

We now argue that $C$ breaks $b$ in the desired sense. For infinitely many $n$'s, 
\begin{enumerate}
\item[(1)]
$
\pr{C(\mathbf{y_i})=1} =
\pr{D'(O \: \mathbf{y_{i}} \:
y_{i+1} \ldots y_m, i-1) = 1} \geq
1 / (m \cdot |D|)
$, where $m \cdot |D|$ is sub-exponential in $n$;

\item[(2)]
$
\pr{C(b(\mathbf{x_i}))=1} =
\pr{D'(O \: b(\mathbf{x_i}) \:
y_{i+1} \ldots y_m, i-1) = 1} =
0
$,  because
$y_{i+1} \ldots y_m \in S_2$ implies there is no $x_1, \ldots, x_i$ such that
$D(b(x_1) \ldots b(x_i) y_{i+1} \ldots y_m) = 1$. 
\end{enumerate}
Since $C$ breaks $b$ in the above sense, we reach a contradiction with the assumption that  $b$ is a demi-bit.
\end{proof}
The condition $c < 1$ in \Cref{alg:stretch} guarantees $n = N^{1-c}$ is polynomially related to $N$ and an infinite monotone sequence of $N$ yields an infinite monotone sequence of $n$.

A key step that makes this proof work is that nondeterministically guessing seeds $x_1, \ldots, x_{i-1}$ in $b(x_1) \ldots b(x_{i-1}) z_{i}$ preserves the ``randomness-structure'' of $b(x_1) \ldots b(x_{i-1}) z_{i}$, in the sense that: when $z_{i} = b(\cdot)$ is pseudorandom, the nondeterministic guess preserves the form $b(\cdot) \ldots b(\cdot) b(\cdot)$ (i.e., $i$ equal-length pseudorandom chunks); and when $z_{i} = y$ is truly random, it preserves the form $b(\cdot) \ldots b(\cdot) y$ (i.e., $i-1$ equal-length pseudorandom chunks followed by a truly random chunk $y$ of the same length). 
For common stretching algorithms that produce exponentially many new bits (e.g., recursively applying a one-bit generator), it is unclear how to use nondeterminism in a way that respects the ``randomness-structure'' of a given string.
Nevertheless, the new proof technique may inspire researchers to further explore the stretchability of demi-bits. 
On the other hand, the fact could also be that there is a specific demi-bit which cannot be stretched to exponentially many demi-bits by the standard stretching algorithms which are applied to super-bits and strong PRGs.

\subsection{Applications in average-case complexity}\label{sec:apps} 
In this section we show that \Cref{thm:alg} implies an equivalence between two different parametric regimes of  zero-error average-case hardness of time-bounded Kolmogorov complexity against \NP/\poly\ machines.

We need the following: a \textbf{\emph{hitting set generator} \emph{against a  class 
of decision problems}} $\mathcal{C}\subseteq 2^{\{0,1\}^{N}}$ is a function
$g: \{0,1\}^n \rightarrow \{0,1\}^{N}$, for $n<N$, such that the image of $g$ \emph{hits} (namely, intersects) every dense enough set $A$ in $\mathcal{C}$ (that is,  $|A|\ge \frac{2^{N}}{N^{{O(1)}}}$). 

As observed by Santhanam \cite{San20}:

\begin{proposition}[\cite{San20}]\label{prop:santh-claim}
Let $n<N$. A hitting set generator $g: \{0,1\}^n \rightarrow \{0,1\}^{N}$ computable in the class $\mathcal D$ against $\NP/\poly$
exists iff there exists a demi-bit $b: \{0,1\}^n \rightarrow \{0,1\}^{N}$ computable in $\mathcal D$ (against $\NP/\poly$).
\end{proposition}  

\begin{proof}
If  $g: \{0,1\}^n \rightarrow \{0,1\}^{N}$ is a hitting set generator against the class $\mathcal{C}$ of decision problems decidable by nondeterministic polynomial-size  circuits, it is also a demi-bit in the sense that no machine $C\in\mathcal{C}$ of polynomial-size  $|C|=N^{O(1)}$ can break $g$, since otherwise $\Pr{C(U_N)=1}\ge 1/N^{O(1)}$ and  $\Pr{C(g(U_n))=1}=0$, contradicting the assumption that $g$ is a hitting set generator against $\mathcal{C}$. Conversely, if $b$ is a demi-bit, than it is also a hitting set generator against $\mathcal{C}$, because if a circuit $C$ in $\mathcal{C}$ outputs $1$ to a dense enough set of inputs it must also output $1$ on a string in the image of $b$, or else $C$ would break the demi-bit.
\end{proof}


Note that~\Cref{alg:stretch} applies also to demi-bits computable in \emph{uniform} polynomial-time (the stretching algorithm is uniform, assuming the original demi-bit is, since it simply applies the demi-bit on different parts of the input). This is important for us, since to talk about Kolmogorov complexity we need machines to be of fixed size, even when the input length changes. Notice, on the other hand, that the proof that the stretching algorithm preserves its hardness necessitates that the adversary $D$ is \emph{non}-uniform (this  is the reason  in \Cref{thm:ac} we  work against \NP/\poly\ adversaries).  
 
\smallskip 


We define the \textbf{$t$-\textit{bounded Kolmogorov complexity of string $x$}}, denoted $\kolm^t(x)$, to be the minimal length of a string $D$ such that the universal Turing machine $U(D)$ (we fix some such universal machine) runs in time at most $t$ and outputs $x$.
See \cite{All2006} for more details about time-bounded Kolmogorov complexity and Definition 9 there for the definition  of time-bounded Kolmogorov complexity of strings  (that definition actually produces the $i$th bit of the string $x$ given an index $i$ and $D$ as inputs to $U$, but this does not change our result).



Recall \Cref{def:kpoly} of the languages \kolmts{s} and \kpolyp.

  

We also need to define precisely the concept of zero-error average-case hardness against the class \NP/\poly\ (equivalently, nondeterministic circuits as in \Cref{def:nondeterministic-circuit}). 
%

\begin{definition}[Zero-error average-case hardness against \NP/\poly]\label{def:zeac-hardness}
We say that a language $L\in\bits^*$ is {\bf zero-error average-case easy} for $\NP/\poly$ if there is an $\NP/\poly$ machine for which all the following hold: (i) every computation-path terminates with either a Yes, No or Don't-Know state; (ii) for a given input $x$\ no two distinct computation-paths terminates with both Yes and No; (iii) we say that the machine answers Yes (No)  on input $x$ if there \emph{exists} a computation-path terminating in Yes (resp.~No) given $x$; (iv) otherwise (namely, all computation-paths given input $x$ terminate in Don't-Know) we say that the machine \emph{does not know} the answer for $x$; (v) the machine never makes a mistake when answering Yes or No (on the other hand, it can answer Don't-Know on either members of $L$ or non-members); and (vi) the machine (correctly) answers Yes or No on at least a polynomial fraction of inputs in $L$ (i.e., at least $2^n/n^c$ of strings in $L\cap\bits^n$, for every sufficiently large $n$ and for some fixed constant $c$ independent of $n$). If $L$ is not zero-error average-case easy for \NP/\poly\ we say that $L$ is {\bf zero-error average-case hard} against $\NP/\poly$.
\end{definition}
 
%
%


The main equivalence is the following:
\begin{theorem}[Equivalence for average-case time-bounded Kolmogorov complexity]\label{thm:ac}
\mbox{$\kpolys{n-O(1)}$} is zero-error average-case hard against $\NP/\poly$ machines iff \mbox{$\kpolys{n-o(n)}$} is zero-error average-case hard against $\NP/\poly$ machines.
\end{theorem}

\begin{proof}
By \Cref{thm:alg} and the equivalence of (uniform polytime computable) demi-bits and hitting set generators (HSGs) against \NP/\poly\ (\Cref{prop:santh-claim}), it suffices to show that for a size function $s:\N\to\N$ with $s(n)<n-O(1)$, $\kpolys{s(n)}$ is zero-error average-case hard against \NP/\poly\ machines iff there is a HSG computable in uniform polytime $H:\bits^{s(n)}\to\bits^n$ against \NP/\poly.

\para{($\Longleftarrow$)} 
Assume that $H:\bits^{s(n)}\to\bits^n$ is a HSG, computable in uniform poly-time, against \NP/\poly. Then, for every constant $k$ (independent of $n$) and sufficiently large $n$, $\text{Im}(H)\cap\bits^n$ intersects all \NP/\poly-computable sets $A_n\subseteq\bits^n $ for which $|A_n|\ge 2^n/n^k$ for all $n$. We need to show that for every constant $c$, $\kolm^{n^c}[s(n)]\subseteq\bits^n$ is zero-error average-case hard for \NP/\poly. We show that there is no \NP/\poly\ machine that answers (correctly) one of Yes or No answers on at least $2^n/n^k$ input strings from $\bits^n$, and on the rest input strings in $\bits^n$ answers Don't-Know, and moreover makes no mistakes. Assume otherwise, then there is an \NP/\poly\ machine that answers (correctly) No for at least $2^n/n^{O(k)}$ input strings from $\bits^n$; this is because most input strings do not have short time-bounded Kolmogorov complexity, that is, a polynomial fraction of the inputs $x\in\bits^n$, for every $n$, are not in \kpoly$[s(n)]$, for $s(n)$ between $|x|^\epsilon$ and $|x|$ (for a constant $0<\epsilon<1$; see \cite[Section 2.6]{All2006} and references therein).
Hence, there is an \NP/\poly/ machine that accepts (correctly) at least $2^n/n^{O(k)}$ ``hard strings'' from $\bits^n$ (namely, strings not in $\kpolys{n}$), and rejects all other strings in $\bits^n$: in particular, the \NP/\poly/ machine guesses the witness the input string being hard, and if the witness is correct it accepts, and otherwise it rejects. 

We thus get a contradiction to $H$ being a HSG against \NP:  there is a dense \NP-language containing at least $2^n/n^{O(k)}$ ``hard strings'' from $\bits^n$. But if $D$ is such an \NP\ machine for this language, then $D$ breaks the HSG $H$: for every string in $\textrm{Im}(H)$ the machine $D$ Rejects, since it has a small $\kpoly$ complexity by the assumption that $H$ is computable in uniform poly-time (in other words, every string $x$ of length $n$ in $\textrm{Im}(H)$ is such that $x\in\kolm^{n^r}[s(n)+O(1)]$, for some constant $r$ independent of $n$, by assumption that $H$ is computable in uniform poly-time).
Hence, $H$ does not hit the dense \NP/\poly-set defined by $D$, a contradiction.

\para{($\Longrightarrow$)} We assume that  $\kpolys{s(n)}$ is zero-error average-case hard against \NP/\poly. Let $H:\bits^{s(n)}\to\bits^n$ be a mapping defined so that the input $x\in\bits^{s(n)}$ is fed into a universal Turing machine to be ran in time $n^k$, for some constant $k$ independent of $n$, and the output of the universal machine is a string of length $n$ (or the string $0\cdots0$ of $n$ zeros if the algorithm does not terminate after $n^k$ steps). We show that $H$ is a HSG against \NP/\poly~(computable in uniform $n^{O(1)}$-time, by assumption).

Assume by way of contradiction that $H$ is not a HSG against \NP/\poly. Then, there is an \NP/\poly\ machine $D$ that accepts at least $2^n/n^{O(1)}$ strings in $\bits^n$, but rejects every string in $\textrm{Im}(H)$.
Therefore, there is an \NP/\poly\ machine $D$ that correctly accepts at least $2^n/n^{O(1)}$ strings $x\in\bits^n$ with $x\not\in\kpoly[s(|x|)]$, and rejects all other strings in $\bits^n$. This contradicts our assumption, because we can construct an \NP/\poly\ machine $D'$ that zero-error decides on average $\kpoly[s(|x|)]$: in $D$ simply replace an Accept state with a Yes state, and a Reject state with a Don't-Know state. 
\end{proof}

\subsection{Application to proof complexity generators}
\begin{corollary}[Stretching proof complexity generators]\label{cor:proof-complexity-generators}
Let $b:\{0,1\}^n\to\{0,1\}^{n+1}$ be a demi-bit computable in \ppoly. Let $0<c<1$ be a constant and $\ell=n+n^c$. Then, there is a proof complexity generator  $g:\{0,1\}^n\to\{0,1\}^{\ell}$  in \ppoly, such that for every  propositional proof system, with probability at least  $1-\frac{1}{\ell^{\omega(1)}}$ over the choice of $r\in\bits^{\ell}$, there are no $\poly(n)$-size proofs of the tautology $\tau(g)_r$.
\end{corollary}

\begin{proof}
By \autoref{thm:alg}, we can stretch the demi-bit $b$ to yield new demi-bits $g:\bits^n\to\bits^{\ell}$. By assumption, $b$ is computable in \ppoly, and by \Cref{alg:stretch} (namely, the algorithm that stretches the demi-bit, which by inspection involves  only applications of the original demi-bit function on different sub-parts of the seed) $g$ is also in \ppoly. 

Assume by way of contradiction that there is a constant $k$ and a  propositional proof system $R$ that admits $\poly(n)$-size proofs of the tautology $\tau(g)_r$, for all $r$ in the set $S\subseteq\bits^\ell\setminus\textup{Im}(g)$ where $|S| \ge 2^\ell/\ell^{k}$ (namely,  it is \emph{not} true that the tautology $\tau(g)_r$ does \emph{not} have polynomial-size $R$-proofs for all $r$  in $\bits^\ell\setminus ({\textup{Im}}(g)\uplus\ T)$ for some $T$ with $|T|\ge 2^\ell/\ell^{\omega(1)}$; note that $|\bits^\ell\setminus({\textup{Im}}(g)\uplus\ T)|= 2^\ell(1-1/\ell^{\omega(1)}-1/2^{\ell-n})=2^\ell(1-1/\ell^{\omega(1)})$). 


Since $g$ is computable in $\ppoly$, the tautology $\tau(g)_r$ is also of size $\poly(n)$. Let $D$ be a nondeterministic circuit that gets an input $r\in S$, ``guesses'' an $R$-proof of $\tau(g)_r$ and verifies it is a correct proof. Then $D$ is of size $n^{O(1)}$, and we have
\begin{equation}\label{eq:ov:tag-2B}
    \underset{y \in \{0, 1\}^{\ell}}{\mathbb{P}} [D(y) = 1] \geq \frac{1}{\ell^k}
    \text{\quad and \ }
    \underset{x \in \{0, 1\}^n}{\mathbb{P}} [D(g(x)) = 1] = 0,
\end{equation}
which contradicts our assumption that $g$ are demi-bits.
\end{proof}

\section{Nondeterministic predictability}\label{sec:predict}
In \Cref{sec:pred}, we review the notion of predictability in the deterministic setting. 
In \Cref{sec:nd-pred}, we introduce new notions of predictability beyond the deterministic setting and study nondeterministic hardness from a predictability aspect.

In this and the next section, we will use a slightly modified definition of super-bits, which is analogous to \Cref{def:strong PRG} in the deterministic setting: a generator $g : \bi{n} \to \bi{m(n)}$ is called super-bits if for every $D$ in \nppoly, every polynomial $p$, and all sufficiently large $n$'s,
\begin{displaymath}
\pr{D(U_{m(n)})=1}
-
\pr{D(g(U_n))=1}
<
1/p(n).
\end{displaymath}
Namely, $g$ is super-bits if $g$ is safe against all \nppoly-distinguishers (the difference from the original definition of super-bits is that security is defined against all \emph{sub-exponential}-size nondeterministic circuits, while here we define it against all polynomial-size nondeterministic circuits).

The contribution of this section provides  progress in our understanding of the nondeterministic hardness of generators.


\subsection{Basic deterministic predictability} 
\label{sec:pred}
We first review the notion of predictability in the deterministic setting and the equivalence between deterministic unpredictability and super-polynomial hardness of generators. The results reviewed in \Cref{sec:pred} are adapted from \cite{txt}.

\begin{definition}[Probability ensembles]
A \textbf{probability ensemble} (or \textbf{ensemble} for short) is an infinite sequence of random variables $(Z_n)_{n \in \bN}$. Each $Z_n$ ranges over $\bi{l(n)}$, where $l(n)$ is polynomially related to $n$ (i.e., there is a polynomial $p$ such that for every $n$ it holds that $l(n) \leq p(n)$ and $p(l(n)) \geq n$.
\end{definition}

We say an ensemble is polynomially generated if there is $g \in P/poly$ such that $g(U_n) = Z_n$. In this paper, we are only interested in polynomially generated ensembles because they are easy to generate to imitate other probability ensembles (e.g., the uniform ensemble, as we will see in the next definition) from a cryptography perspective. Thus, every ensemble we consider from now on will be implicitly assumed to be polynomially generated if not specified otherwise.

\begin{definition}[Strong-pseudorandom ensembles]
The probability ensemble $(Z_n)_{n \in \bN}$ is \textbf{strongly pseudorandom} if for every algorithm $D$ in \ppoly{}, every polynomial $p$, and all sufficiently large $n$'s,
\[
|\bP_{U_{m(n)}}[D(U_{m(n)}, 1^n) = 1] - \bP_{Z_n}[D(Z_n, 1^n) = 1]|
\leq
1/p(n),
\]
where $m = |Z_n|$ and $U_m$ is the uniform distribution
\end{definition}
The input $1^n$ allows $D$ is run poly-time in $n$ and informs $D$ of the value $n$. But for the sake of notation simplicity, $1^n$ is often omitted and will be omitted from now on. 

We note that if $g$ is a strong PRG, then $g(U_n)$ is a strong-pseudorandom ensemble. We will see, for $g(U_n)$, being strongly pseudorandom is equivalent to being unpredictable.

\begin{definition}[(Deterministic) predictability]
An ensemble $(Z_n)_{n \in \bN}$ is called \textbf{predictable} or \textbf{P/poly-predictable} is there exist an algorithm $\McA$ in \ppoly{}, a polynomial $p$, infinitely many $n$'s, and an $i(n) < |Z_n|$ for each of those $n$'s such that
\[
\pr{
\McA(Z_n[1 \ldots i]) = Z_n[i+1]}
\geq
\frac{1}{2} + 1/p(n).
\]
An ensemble $(Z_n)_{n \in \bN}$ is \textbf{unpredictable} if it is not predictable.
\end{definition}

\begin{theorem} \label{det-pred-thm}
An ensemble is strong-pseudorandom if and only if it is unpredictable. 
\end{theorem}

We will see in the next section a nondeterministic variant of the last theorem.


\subsection{Nondeterministic variants}
\label{sec:nd-pred}

In this section, we propose four new notions of unpredictability beyond the deterministic setting and characterise super-hardness of generators based on the new notions.

Before we define nondeterministic predictability, we emphasise an important distinction between a decision problem and a single-bit-output computing problem in the nondeterministic setting.
We say a nondeterministic algorithm $\McA$ is a \textbf{function-computing} algorithm, 
if for every input $x \in \bi{n}$,
every computation branch yields one of $\{0, 1, \bot \}$, in which $\bot$ indicates a failure,
and there is always a computation branch yielding $0$ or $1$.
$\McA(x) = c$ for some $c \in \bi{}$ if, on input $x$, every computation branch either yields $c$ or $\bot$; otherwise $\McA(x) = \bot$.
Hence, if $\McA(x) = \bot$, then there are a computation branch yielding $0$ and another computation branch yielding $1$.
We say $\McA$ is total if $\McA(x) \in \bi{}$ for every $x \in \bi{n}$. Obviously, total function-computing algorithms constitute a subclass of function-computing algorithms.
If we restrict the algorithms to be total for defining $\mathbf{\cap}$-unpredictability in \Cref{def:nondeterministic-predictability},  the diagram shown in \Cref{sum:unpredict} at the end of this section remains unchanged. 
The advantage of not putting this restriction is that the property \mbox{$\mathbf{\cap}$-unpredictability} becomes slightly stronger and thus, with this restriction we achieve a tighter ``sandwiched'' characterisation of nondeterministic hardness in \Cref{sum:unpredict}.

We also remark that a decision problem is in $\NP/\poly \cap \coNP/\poly$ if and only if it is decidable by a (total) nondeterministic polynomial-size function-computing algorithm.


\begin{definition}[Nondeterministic predictability]\label{def:nondeterministic-predictability}\
\begin{description}

\item[$\mathbf{\NP/\poly}$-predictability.] 
An ensemble $(Z_n)_{n \in \bN}$ is \nppoly-predictable if there exist an algorithm $\McA$ in \nppoly{}, a polynomial $p$, infinitely many $n$'s, and an $i(n) < |Z_n|$ for each of those $n$'s such that
\[
\pr{
\McA(Z_n[1 \ldots i]) = Z_n[i+1]}
\geq
1/2 + 1/p(n).
\]
An ensemble $(Z_n)_{n \in \bN}$ is \nppoly-unpredictable if it is not \nppoly-predictable.

\item[$\mathbf{\coNP/\poly}$-predictability.]
An ensemble $(Z_n)_{n \in \bN}$ is \conppoly-predictable if there exist an algorithm $\McA$ in \conppoly{}, a polynomial $p$, infinitely many $n$'s, and an $i(n) < |Z_n|$ for each of those $n$'s such that
\[
\pr{
\McA(Z_n[1 \ldots i]) = Z_n[i+1]}
\geq
1/2 + 1/p(n).
\]
An ensemble $(Z_n)_{n \in \bN}$ is \conppoly-unpredictable if it is not \conppoly-predictable.

\item[$\mathbf{\cup}$-predictability.]
An ensemble $(Z_n)_{n \in \bN}$ is called $\cup$-predictable (a short for $(\NP/\poly \cup \coNP/\poly)$-predictable) if it is \nppoly-predictable or \conppoly-predictable. 
An ensemble $(Z_n)_{n \in \bN}$ is $\cup$-unpredictable if it is not $\cup$-predictable.

\item[$\mathbf{\cap}$-predictability.]
An ensemble $(Z_n)_{n \in \bN}$ is called $\cap$-predictable (a short for $(\NP/\poly \cap \coNP/\poly)$-predictable) if there exist a poly-time nondeterministic \textbf{function-computing} algorithm $\McA$, a polynomial $p$, infinitely many $n$'s, and an $i(n) < |Z_n|$ for each of those $n$'s such that
\[
\pr{
\McA(Z_n[1 \ldots i]) = Z_n[i+1]} 
\geq
1/2 + 1/p(n).
\]
An ensemble $(Z_n)_{n \in \bN}$ is $\cap$-unpredictable if it is not $\cap$-predictable.
\end{description}
\end{definition}

\begin{remark*} A reader should not mistake the definition of being $\cap$-predictable as being \nppoly-predictable $\land$ \conppoly-predictable. In fact, being $\cap$-predictable is a stronger assumption than being \nppoly-predictable $\land$ \conppoly-predictable (see Lemma \Cref{pred-lemma}). 
In other words, 
for an ensemble $(Z_n)$, being $\cap$-unpredictable is a weaker property than being \nppoly-unpredictable $\lor$ \conppoly-unpredictable.
Nevertheless, being $\cap$-unpredictable is still a non-trivial property because it is stronger than being \ppoly-unpredictable (a deterministic algorithm is a total function-computing algorithm).
\end{remark*}

\begin{lemma} \label{pred-lemma}
If an ensemble $(Z_n)$ is $\cap$-predictable, it is both \nppoly-predictable and \conppoly-predictable.
\end{lemma}
\begin{proof}
Suppose there exist a poly-time nondeterministic \textbf{function-computing} algorithm $\McA$, a polynomial $p$, infinitely many $n$'s, and an $i(n) < |Z_n|$ for each of those $n$'s such that
\[
\pr{\McA(Z_n[1 \ldots i]) = Z_n[i+1]} 
\geq
1/2 + 1/p(n).
\]

We define $\McA_1$ in \nppoly{} and $\McA_0$ in \conppoly{} as follows:
\begin{quote}
    Given input $Y \in \bi{i}$, $\McA_j (j=0,1)$ mimics $\McA$ on input $Y$ to get an output bit $c$. If $c = \bot$, $\McA_j$ outputs $1-j$, and otherwise $\McA_j$ outputs $c$.
\end{quote}
By the definition of $\McA_j$, if $\McA(Y) \in \bi{}$, $\McA_j(Y) = \McA(Y)$. 
Thus, for $j \in \bi{}$, we have
\begin{align*}
  \pr{\McA_j(Z_n[1 \ldots i]) = Z_n[i+1]} 
  & \geq
  \pr{\McA_j(Z_n[1 \ldots i]) = Z_n[i+1] \land \McA(Z_n[1 \ldots i]) \in \bi{}} \\
  & =
  \pr{\McA(Z_n[1 \ldots i]) = Z_n[i+1] \land \McA(Z_n[1 \ldots i]) \in \bi{}} \\
  & =
  \pr{\McA(Z_n[1 \ldots i]) = Z_n[i+1]} \\
  & \geq 
  1/2 + 1/p(n).  
\end{align*}
\end{proof}

\begin{proposition}\label{pred-prop1}
If $g : \bi{n} \to \bi{m(n)}$ is super-bit(s), then $g(U_n)$ is $\cap$-unpredictable.
\end{proposition}
\begin{proof}
Suppose, for a contradiction, $g(U_n)$ is $\cap$-predictable.
there exist a poly-time nondeterministic \textbf{function-computing} algorithm $\McA$, a polynomial $p$, infinitely many $n$'s, and an $i(n) < |g(U_n)|$ for each of those $n$'s such that
\[
\pr{\McA(Z_i) = z} 
\geq
1/2 + 1/p(n),
\]
where $Z_i = g(U_n)[1 \ldots i]$ and $z = g(U_n)[i+1]$.

\vspace{0.6em}

We construct nondeterministic algorithm $D$ to distinguish $U_m$ from $g(U_n)$:
\begin{quote}
    Given input $Y \in \bi{m}$, $D$ runs $\McA$ on $Y[1 \dots i]$ to obtain an output bit $c$. If $c = \bot$, $D$ outputs $0$. When $c \in \bi{}$, if $c \neq Y[i+1]$, $D$ outputs $1$; else, $D$ outputs $0$.
\end{quote}

For a fixed $n$, we let $f = \pr{\McA(U_i) = \bot}$ and use $b$ to denote a random bit. 
Recall that $\McA(U_i) = \bot$ implies the computation tree of $\McA$ on $U_i$ has both $0$ and $1$ as leaves.
Then by the definition of $D$, we have:
\begin{align*}
   \pr{D(U_m) = 1} 
    &=
   \pr{\McA(U_i) = \bot} +
   \pr{\McA(U_i) \in \bi{} \land \McA(U_i) \neq b} \\
    &= 
   \pr{\McA(U_i) = \bot} +
   \pr{\McA(U_i) \in \bi{}} \cdot
   \pr{\McA(U_i) \neq b | \McA(U_i) \in \bi{}} \\
    &=
   f + (1-f) \cdot 1/2 \\
    &\geq
   1/2.
\end{align*}
On the other hand:
\begin{align*}
    \pr{D(g(U_n)) = 1}
   &=
   \pr{\McA(Z_i) = \bot} +
   \pr{\McA(Z_i) \in \bi{} \land \McA(Z_i) \neq z} \\
   &=
   \pr{\McA(Z_i) = \bot} +
   (
   \pr{\McA(Z_i) \in \bi{}}
   -
   \pr{\McA(Z_i) \in \bi{} \land \McA(Z_i) = z} 
   ) \\
   &=
   \pr{\McA(Z_i) = \bot} +
   \pr{\McA(Z_i) \in \bi{}}
   -
   \pr{\McA(Z_i) = z} \\
   &=
   1 - \pr{\McA(Z_i) = z} \\
   &\leq
   1/2 - 1/p(n).
\end{align*}
Therefore,
\begin{align*}
    \pr{D(U_m) = 1} - 
    \pr{D(g(U_n)) = 1}
    \geq
    1/p(n).
\end{align*}
\end{proof}

\begin{proposition}\label{pred-prop2}
If $g(U_n)$ is $\cup$-unpredictable (i.e., \nppoly-unpredictable $\land$ \conppoly-unpredictable), where $g: \bi{n} \to \bi{m(n)} \in P/poly$ and $m(n) > n$, then $g$ is super-bit(s).
\end{proposition}
\begin{proof}
Suppose, for a contradiction, g is not super-bits. For any fixed $n$, define hybrids $H_i = g(U_n)[0 \dots i] \cdot U_m[i+1 \dots m]$ and denote $Z_i = g(U_n)[1 \dots i]$ and $z_i = g(U_n)[i]$. Then, by the hybrid argument, there exist an algorithm $D$ in \nppoly{}, a polynomial $p$, infinitely many $n$'s, and an $i$ for each of these $n$'s such that 
\[
\pr{D(H_i) = 1} - \pr{D(H_{i+1}) = 1} \geq 1/p(n).
\]
Therefore, for each such $n$, there is a fixed string $w$ such that
\begin{equation*}
    \pr{D(Z_i b, w) = 1} - \pr{D(Z_{i+1}, w) = 1} \geq 1/p(n),
    \tag{$\ast$}
\end{equation*}
where we use $b$ to denote a random bit. We may omit writing $w$ from now on.

We construct an algorithm $\McA_1$ in \nppoly{} and $\McA_2$ in \conppoly{} to predict $z_{i+1}$ based $Z_i$ as follows:
\begin{adjustwidth}{1.2cm}{1.2cm}
    Given $Y_i \in \bi{i}$ as input, $\McA_1$ does whatever $D$ does on $Y_i0w$; $\McA_2$ does whatever $D$ does on $Y_i1w$ but then flip the bit got.
\end{adjustwidth}
The definitions of $\McA_1, \McA_2$ imply that $\McA_1(Y_i) = D(Y_i0)$ and $\McA_2(Y_i) = \overline{D(Y_i1)}$ ($w$ omitted).
Hence,
\begin{align*}
    (1)
    &:=
    \pr{\McA_1(Z_i) = z_{i+1}} \\
    &=
    \pr{D(Z_i0) = z_{i+1}} \\
    &=
    \pr{D(Z_i0) = 0 \land z_{i+1} = 0} +
    \pr{D(Z_i0) = 1 \land z_{i+1} = 1} \\ 
    &=
    \pr{D(Z_i z_{i+1}) = 0 \land z_{i+1} = 0} +
    \pr{D(Z_i \overline{z_{i+1}}) = 1 \land z_{i+1} = 1},
\end{align*}
and
\begin{align*}
    (2)
    &:=
    \pr{\McA_2(Z_i) = z_{i+1}} \\
    &=
    \pr{D(Z_i1) \neq z_{i+1}} \\
    &=
    \pr{D(Z_i1) = 0 \land z_{i+1} = 1} +
    \pr{D(Z_i1) = 1 \land z_{i+1} = 0} \\ 
    &=
    \pr{D(Z_i z_{i+1}) = 0 \land z_{i+1} = 1} +
    \pr{D(Z_i \overline{z_{i+1}}) = 1 \land z_{i+1} = 0}.
\end{align*}
Therefore, 
\begin{align*}
    (1) + (2)
    &=
    \pr{D(Z_i z_{i+1}) = 0} + 
    \pr{D(Z_i \overline{z_{i+1}}) = 1}.
\end{align*}
As 
\[
\pr{D(Z_i z_{i+1}) = 0} = 1 - \pr{D(Z_{i+1}) = 1}
\] 
and
\begin{align*}
    2\pr{D(Z_i,b) = 1} 
    &= 
    2\pr{D(Z_i,b) = 1 \land b = z_{i+1}} + 
    2\pr{D(Z_i,b) = 1 \land b \neq z_{i+1}} \\
    &=
    2\pr{b = z_{i+1}}\pr{D(Z_i,b) = 1 | b = z_{i+1}} + 
    2\pr{b \neq z_{i+1}}\pr{D(Z_i,b) = 1 | b \neq z_{i+1}} \\
    &=
    \pr{D(Z_i,z_{i+1}) = 1} +
    \pr{D(Z_i,\overline{z_{i+1}}) = 1},
\end{align*}
\begin{align*}
    (1) + (2)
    &= 
    (1 - \pr{D(Z_{i+1}) = 1}) 
    + 
    (2\pr{D(Z_ib) = 1} - \pr{D(Z_{i+1}) = 1}) \\
    &=
    1 + 2(\pr{D(Z_ib) = 1} - \pr{D(Z_{i+1}) = 1})\\
    &\geq
    1 + 2/p(n) \text{ by ($\ast$).}
\end{align*}

Hence, either $(1) = \pr{\McA_1(Z_i) = z_{i+1}} \geq 1/2 + 1/p(n)$ for infinitely many $n$'s or $(2) = \pr{\McA_2(Z_i) = z_{i+1}} \geq 1/2 + 1/p(n)$ for infinitely many $n$'s. That is, the ensemble $(Z_n)$ is either \nppoly-predictable or \conppoly-predictable.
\end{proof}

\noindent

\begin{summary}\label{sum:unpredict}
We summarise \Cref{pred-lemma}, \Cref{pred-prop1}, and \Cref{pred-prop2} together as the following chain of inequalities. Here, $A \leq B$ means, if an ensemble $g(U_n)$ has property $B$, then it also has property $A$ (see introduction \Cref{sec:overv:big-formula} for a more pictorial version of these inequalities):
\begin{align}
&\ppoly\text{-unpredictability}
\leq
\cap\text{-unpredictability}\\
 &  \cap\text{-unpredictability}
    \leq
    \text{super-polynomial nondeterministic hardness}
    \leq
    \cup\text{-unpredictability} 
\\
&    \cap\text{-unpredictability}
    \leq
    \NP/\poly\text{-unpredictability} 
    \leq
    \cup\text{-unpredictability}.
\\
&    \cap\text{-unpredictability}
    \leq
    \coNP/\poly\text{-unpredictability}
    \leq
    \cup\text{-unpredictability}.
\end{align}
\end{summary}
In contrast to \Cref{det-pred-thm}, which is a precise characterization of standard hardness from a predictability perspective, we have so far only obtained inaccurate characterizations of super-hardness in the nondeterministic setting. I am inclined to the viewpoint that we might be unable to obtain an exact characterization of super-hardness in terms of predictability, because we will see in \Cref{sec:super-core}, super-hardness is not just about algorithm (e.g., algorithms used to witness randomness or used to predict) behaviour on the range elements (i.e., $y$ such that $\pr{Z_n = y} > 0$) but may also concern behaviour on the non-range elements (i.e., $y$ such that $\pr{Z_n = y} = 0$). In other words, it is very possible that at least some of the inequalities above are strict. 
\begin{open problem*}
Could the inequalities in \textbf{Summary} 
 be further refined or classified? For example, is there any relation between super-hardness and \nppoly-unpredictable
$\lor$ 
\conppoly-unpredictable?
\end{open problem*}


\section{Super-core predicates}
\label{sec:hard-core}

In the next subsection  we review the concepts of one-way functions and hard-core predicates as well as their known connections to strong PRGs.
In \Cref{sec:super-core} we introduce the concept of a  super-core predicate, and investigate  its connections to  super-bits.
This provides a step forward for suggesting a sensible definition of one-way functions in the nondeterministic setting.

\subsection{One-way functions and hard-core predicates}
\label{sec:One-way functions and hard-core predicates}

We start by  reviewing the standard concepts of one-way functions and hard-core predicates and the equivalence between the existence of a strong PRG and the existence of a hard-core of some one-way function (see  \cite{txt} for more details).

\begin{definition}[One-way functions]
A function $f: \bi{*} \to \bi{*}$ in \ppoly{} is called \textbf{one-way} if for every $\McA$ in \ppoly{}, every polynomial $p(\cdot)$, and all sufficiently large $n$'s,
\[ 
\Pr[x \in \bi{n}]{\McA(f(x), 1^n) \in f^{-1}(f(x))} 
< 
\frac{1}{p(n)}. 
\]
\end{definition}
The input $1^n$ is for technical reason. It allows the algorithm $\McA$ to run in polynomial time in $n = |x|$, which is important when $f$ dramatically shrinks its input (e.g., when $|f(x)| = O(log|x|)$). When the auxiliary input $1^n$ is not necessary (e.g., when $f$ is length-preserving), we may omit it. Intuitively, a function $f$ in \ppoly{} is one-way if it is ``typically" hard to invert, when the probability is taken over the input distribution, for all efficient algorithms.

It is known that the existence of one-way functions and the existence of strong PRGs are equivalent:
\begin{theorem}\label{equiv0}
One-way functions exist if and only if strong PRGs exist.
\end{theorem}

Since we can construct a length-preserving one-way function from an arbitrary one-way function, \autoref{equiv0} can  alternatively be stated as:
\begin{theorem}\label{equiv1}
Length-preserving one-way functions exist if and only if strong PRGs exist.
\end{theorem}

The converse implication  of \autoref{equiv0} is rather straightforward, while the forward implication, in fact, relies on a concept closely related to one-way functions, called \emph{hard-core predicates}. Even when assuming that a  hard-core predicate can produce one more bit from a seed $x$, the known proof of the forward implication in \autoref{equiv1} is still rather involved. 
\begin{definition}[Hard-core predicates]
A predicate $b: \bi{*} \to \bi{}$ in \ppoly\ is called a \textbf{hard-core} of a function $f$ if 
there do not exist an algorithm $\McA$ in \ppoly, a polynomial $p(\cdot)$, and infinitely many $n$'s such that
\[ 
\Pr[x \in \bi{n}]{\McA(f(x), 1^n) = b(x)} 
\geq
\frac{1}{2} + \frac{1}{p(n)}. \]
\end{definition}
In other words, $b$ is a hard-core of $f$ if it is safe against (in terms of prediction) all \ppoly{} algorithms, which may be due to an information loss of $f$ (e.g, $b(x) = x[n]$ is a hard-core of $f(x) = x[1 \ldots (n-1)]$) or to the difficulty of inverting $f$.

If $b$ is a hard-core of any $f$, then $\pr{b(x) = 0} \approx \pr{b(x) = 1} \approx 1/2$ (otherwise, we can predict $b$ by outputting the constant $argmax_{b \in \bi{}}(\pr{b(x) = b})$).

One-way functions and hard-core predicates are ``paired" by the construction of a ``generic" hard-core:
\begin{theorem}\label{hard-core one-way}
For any one-way function $f$, the inner-product mod 2 of $x$ and $y$, denoted as $\langle x,y \rangle$, is a hard-core of $f'(x,y):= (f(x), y)$.
\end{theorem}
In particular, if $f$ is length-preserving, so is $f'$. Indeed, in the theorem, we should also define $f'$ and $b$ on the odd length inputs $(x,y,b)$, where $|x|=|y|$ and $b$ is a single bit, but this case is often omitted as we can trivially ``ignore" the last bit (use the same idea as in \Cref{i.o.}) and define $f'(x,y,b) = (f(x),y,b)$ and $b(x,y,b) = \langle x,y \rangle$.

This theorem states that every one-way function $f$ ``essentially'' has the same hard-core predicate (where here it's not $f$ itself that has the hard-core predicate, rather $f'$ as above). It is an easy exercise to show that there does not exist a predicate $b$ that is hard core for every one-way function.

Therefore, \Cref{equiv1} can be reformulated as:
\begin{theorem}\label{equiv2}
There is a hard-core $b$ of some length-preserving one-way function if and only if there is a strong PRG $g$.
\end{theorem}

We will develop a nondeterministic variant (\Cref{equiv4}) of the last theorem in the following section.


\subsection{Nondeterministic variants}
\label{sec:super-core}

In this section,
we introduce the concept of  super-core predicates and investigate its connection to  super-bits as well as  which functions may or may not have super-cores.

\begin{definition}[Super-core predicates]\label{def:super-core}
\label{def:supe-core}
A predicate $b: \bi{n} \to \bi{}$ in \ppoly\ is called a \textbf{super-core} of a function $f:\bi{n} \to \bi{m(n)}$ if there do not exist an algorithm $\McA_1$ in \nppoly, an algorithm $\McA_2$ in \conppoly{}, polynomial $p(\cdot)$, and infinitely many $n$'s such that
\[
\bP_{x \in \bi{n}}[\McA_1(f(x), 1^n) = b(x) = 0]
+
\frac{1}{2}\bP_{y \in \bi{m}}[\McA_1(y, 1^n) = 1]
\geq
\frac{1}{2} + \frac{1}{p(n)}
\tag{$\star$}
\]
\centerline{nor}
\[
\bP_{x \in \bi{n}}[\McA_2(f(x), 1^n) = b(x) = 1]
+
\frac{1}{2}\bP_{y \in \bi{m}}[\McA_2(y, 1^n) = 0]
\geq
\frac{1}{2} + \frac{1}{p(n)}\,
\tag{$\diamond$}
.
\]
\end{definition}
%

For convenience, we define some abbreviations for the formulas above (the input $1^n$ is omitted): 
\begin{align*}
\circled{1}(\McA_1) & :=
\bP_{x \in \bi{n}}[\McA_1(f(x)) = b(x) = 0],
\\
~~~\circled{2}(\McA_2) & :=
\bP_{x \in \bi{n}}[\McA_2(f(x)) = b(x) = 1],
\\
\circled{3}(\McA_1) & :=
\tfrac{1}{2}\bP_{y \in \bi{m(n)}}[\McA_1(y) = 1], 
\\
\circled{4}(\McA_2) & :=
\tfrac{1}{2}\bP_{y \in \bi{m(n)}}[\McA_2(y) = 0].
\end{align*}

Our intention here is to come up with a nondeterministic variant of hard-core predicates and to build a relation between the existence of super-bits and the existence of this nondeterministic variant.
A super-core is safe against both nondeterministic and co-nondeterministic predictors in the above-prescribed sense. 
Inequality ($\star$) can be re-written as 
\[
\bP_{x \in \bi{n}}[\McA_1(f(x), 1^n) b(x) = 0]
\geq
\frac{1}{2}\bP_{y \in \bi{m}}[\McA_1(y, 1^n) = 0] + \frac{1}{p(n)},
\]
and splitting the left-hand side in terms of conditional probability on $b(x) = 0$ yields
\[
\bP_{x \in \bi{n}}[b(x) = 0]
\bP_{x \in \bi{n}}[\McA_1(f(x), 1^n) = 0 |
b(x) = 0]
\geq
\frac{1}{2}\bP_{y \in \bi{m}}[\McA_1(y, 1^n) = 0] + \frac{1}{p(n)}.
\]
As we will see in \Cref{lem:super-core implies hard}, a super-core is also a hard-core, which implies 
$\bP_{x \in \bi{n}}[b(x) = 0] \approx 1/2$.
Thus, for an \nppoly-predictor $\McA_1$ to satisfy ($\star$), it has to satisfy the following equivalent condition for some polynomial $p(n)$:
\[
\bP_{x \in \bi{n}}[\McA_1(f(x), 1^n) = 0 |
b(x) = 0]
\geq
\bP_{y \in \bi{m}}[\McA_1(y, 1^n) = 0] + 
\frac{1}{p(n)}.
\]
This condition intuitively means that if the input $y$ given to the \nppoly-predictor $\McA_1$ indeed represents some $f(x)$ such that $b(x) = 0$, the \nppoly-predictor $\McA_1$ has to be more sensitive to detect it by outputting $0$ (than when receiving a random input).
($\diamond$) is dual to ($\star$).
Note that, under this definition, nondeterministic and co-nondeterministic adversaries are unable to satisfy their corresponding inequality in any trivial way (e.g., by outputting the constant $0$ or $1$).

\medskip 

We observe that if $b$ is a super-core of $f$, then $b$ is also a super-core of any ``shortened" $f$. Precisely, if $i: \bN \to \bN$ is such that $i(n) \leq |f(1^n)|$ for every $n$, then $b$ is a super-core of $f'(x) := f(x)[1 \ldots i(|x|)]$ because $f'(x)$ provides less information than $f(x)$.

Just as a super-bit is also a strong PRG, we have:
\begin{lemma}\label{lem:super-core implies hard}
If $b$ is a super-core of function $f: \bi{n} \to \bi{m(n)}$, $b$ is a hard-core of $f$.
\end{lemma}
\begin{proof}
Suppose for a contradiction that $b$ is not a hard-core of $f$. Then there exist $C$ in \ppoly, a polynomial $p$, and infinitely many $n$'s such that
$
\pr{C(f(U_n)) = b(U_n)} \geq 1/2 + 1/p(n).
$
We note $C$ is in both \nppoly{} and \conppoly. 
As 
\[
\circled{1}(C) + \circled{2}(C) +\circled{3}(C) +\circled{4}(C)
=
\pr{C(f(U_n)) = b(U_n)} + 1/2
\geq
1 + 1/p(n),
\]
either 
$\circled{1}(C) + \circled{3}(C) \geq 1/2 + 1/(2p(n))$
or
$\circled{2}(C) + \circled{4}(C) \geq 1/2 + 1/(2p(n))$.
\end{proof}

We now show that we can construct a super-bit from a super-core for some length-preserving function:
\begin{proposition}\label{super-core to bit}
If $b$ is a super-core of function $f: \bi{n} \to \bi{n}$ in \ppoly{}, then $g(x):= f(x)b(x)$ is a super-bit.
\end{proposition}
\begin{proof}
Suppose, for a contradiction, $g$ is not a super-bit. Then there exist a distinguisher $D$ in \nppoly, a polynomial $p$, and infinitely many $n$'s such that:
\[
\pr{D(U_{n+1})=1} - \pr{D(g(U_n))=1}
\geq
1/p(n).
\]

We define algorithm $\McA_1$ in \nppoly{} and algorithm $\McA_2$ in \conppoly{} to predict $b$ as follows:
\begin{adjustwidth}{1.2cm}{1.2cm}
    Assume $Y \in \bi{n}$ is given as input. $\McA_1$ runs $D$ on $Y0$ and accepts when  $D$ accepts. $\McA_2$ runs $D$ on $Y1$ and accepts when $D$ rejects.
\end{adjustwidth}
The definitions of $\McA_1,\McA_2$ means that $\McA_1(Y) = D(Y0)$ and $\McA_2(Y) = \overline{D(Y1)}$. Now, we have:
\begin{align*}
    \circled{1}(\McA_1) +
    \circled{2}(\McA_2)
    &=
    \pr{\McA_1(f(U_n)) = b(U_n) = 0} +
    \pr{\McA_2(f(U_n)) = b(U_n) = 1} \\
    &=
    \pr{D(f(U_n)0) = 0 \land b(U_n) = 0} + 
    \pr{D(f(U_n)1) = 0 \land b(U_n) = 1} \\
    &=
    \pr{D(f(U_n)b(U_n)) = 0 \land b(U_n) = 0} + 
    \pr{D(f(U_n)b(U_n)) = 0 \land b(U_n) = 1} \\
    &=
    \pr{D(f(U_n)b(U_n)) = 0} \\
    &=
    1 - \pr{D(g(U_n)) = 1},
\end{align*}
and
\begin{align*}
    \circled{3}(\McA_1) +
    \circled{4}(\McA_2)
    &=
    \frac{1}{2} \pr{\McA_1(U_n) = 1} +
    \frac{1}{2} \pr{\McA_2(U_n) = 0} \\
    &=
    \frac{1}{2} \pr{D(U_n0) = 1} +
    \frac{1}{2} \pr{D(U_n1) = 1} \\
    &=
    \pr{D(U_{n+1}) = 1}
    .
\end{align*}

Therefore, for infinitely many $n$'s,
\[
\circled{1}(\McA_1) + 
\circled{2}(\McA_2) + 
\circled{3}(\McA_1) + 
\circled{4}(\McA_2)
=
1 + \pr{D(U_{n+1}) = 1} - \pr{D(g(U_n)) = 1}
\geq
1 + 1/p(n),
\]
which implies that either
$\circled{1}(\McA_1) + \circled{3}(\McA_1) \geq 1/2 + 1/(2p(n))$
for infinitely many $n$'s
or
$\circled{2}(\McA_2) + \circled{4}(\McA_2) \geq 1/2 + 1/(2p(n))$
for infinitely many $n$'s, contradicting the assumption that $b$ is a super-core of $f$.
\end{proof}

Before we establish \Cref{super-bit to core}, which is the converse of \Cref{super-core to bit}, we need the following auxiliary lemma:
\begin{lemma}\label{hard-bit to core}
If $g: \bi{n} \to \bi{n+1}$ is a strong PRG and define $f(x)b(x) := g(x)$, where $b(x)$ is the last bit of $g(x)$, then $b$ is a hard-core of $f \in P/poly$.
\end{lemma}
\begin{proof}
Suppose for contradiction, $b$ is not a hard-core of $f$.
Then there exist $\McA$ in \ppoly{}, a polynomial $p$, and infinitely many $n$'s such that
\[
\pr{\McA(f(U_n)) = b(U_n)} 
\geq 1/2 + 1/p(n).
\]

We construct $D$ to break $g$ as follows:
\begin{adjustwidth}{1.2cm}{1.2cm}
    Given $Y \in \bi{n+1}$ as input, $D$ outputs $1$ if and only if $\McA(Y[1..n]) = Y[n+1]$.
\end{adjustwidth}
Let $b$ denote a random bit. Then, for infinitely many $n$'s,
\begin{align*}
    \pr{D(g(U_n))=1} -
    \pr{D(U_{n+1})=1} 
    &=
    \pr{\McA(f(U_n)) = b(U_n)} -
    \pr{b = \McA(U_n)}
    \\&\geq
    1/2 + 1/p(n) - 1/2
    \\&=
    1/p(n).
\end{align*}
\end{proof}

Now, we are able to show:
\begin{proposition}\label{super-bit to core}
If $g: \bi{n} \to \bi{n+1}$ is a super-bit and define $f(x)b(x) := g(x)$, where $b(x)$ is the last bit of $g(x)$, then $b$ is a super-core of $f \in \P/\poly$.
\end{proposition}
\begin{proof}
Suppose, for a contradiction, $b$ is not a super-core of $f$. We present a proof for the case in which there exist an $\McA \in \NP/\poly$, a polynomial $p$, and infinitely many $n$'s such that 
$
\circled{1}(\McA) 
+ 
\circled{3}(\McA)
\geq
1/2 + 1/p(n)
$. The proof for the other case is similar.

We define a distinguisher $D$ to break $g$ as follow:
\begin{adjustwidth}{1.2cm}{1.2cm}
    Given $Y \in \bi{n+1}$ as input, $D$ runs $\McA$ on $Y[1 \ldots n]$ to get one output bit $c$ of $\McA$. $D$ then outputs $1$ if and only if $Y[n+1] = 0$ and $c = 1$. 
\end{adjustwidth}

Then, for infinitely many $n$'s, we have
\begin{align*}
        \phantom{==}
  &  \pr{D(U_{n+1}) = 1} -
    \pr{D(f(U_n)b(U_n)) = 1}
    \\&=
    \pr{\McA(U_n)=1 \land U_1 = 0} -
    \pr{\McA(f(U_n))=1 \land b(U_n) = 0}
    \\&=
    1/2 \cdot \pr{\McA(U_n)=1} -
    (\pr{b(U_n) = 0} - 
     \pr{\McA(f(U_n))=0 \land b(U_n) = 0})
    \\&=
    \circled{1}(\McA) + \circled{3}(\McA) -
    \pr{b(U_n) = 0}
    \\&\geq
    1/2 + 1/p(n) - \pr{b(U_n) = 0}.
\end{align*}

Let $1/s(n) = |\pr{b(U_n) = 0} - 1/2|$. As $g$ is a strong PRG, $b$ is a hard-core of $g$ by \Cref{hard-bit to core}. Thus, for every polynomial $q$ (in particular, for $q(n) = 2p(n)$) and every sufficiently large $n$, $1/s(n) \leq 1/q(n)$. Therefore, for infinitely many $n$'s,
\begin{align*}
    \pr{D(U_{n+1}) = 1} -
    \pr{D(f(U_n)b(U_n)) = 1}
    \geq
    1/2 + 1/p(n) - \pr{b(U_n) = 0}
    \geq
    1/2p(n).
\end{align*}
\end{proof}

By combining \Cref{super-core to bit} and \Cref{super-bit to core}, we establish the following nondeterministic variant of \Cref{equiv2}.
\begin{theorem}\label{equiv3}
There is a super-core $b$ of some length-preserving $f \in \P/\poly$ if and only if there is a super-bit $g$.
\end{theorem}

We say a function $f: \bi{n} \to 
\bi{m(n)}$ is \textbf{\emph{non-shrinking}} if $m(n) \geq n$ for every $n$. Because (1) $b$ is a super-core of some length-preserving function if and only if $b$ is a super-core of some non-shrinking function, and (2) there exists a super-bit if and only if there exist super-bits (i.e., we can stretch super-bits by \cite{superbit}, we can alternatively state \Cref{equiv3} as:
\begin{theorem}\label{equiv4}
There is a super-core of some non-shrinking function in $\P/\poly$ if and only if there are super-bits.
\end{theorem}


Although we can relax the condition ``length-preserving" to ``non-shrinking", the requirement ``non-shrinking" is not redundant because there is an easy construction of a super-core $b$ of some function $f$ which shrinks its input, but whether a super-bit exists is unknown. Define $f(x) = x[1]$ and $b(x) = x[-1]$, then $b$ is a super-core of $f$ as there are only four functions from $\bi{}$ to $\bi{}$, and none of them can predict $b$ given $f(x) \in \bi{}$ in the required sense (indeed, $\circled{1}(c) + \circled{3}(c) = \circled{2}(c) + \circled{4}(c) = 1/2$ for any $c : \bi{} \to \bi{}$).

At this point, we may want to understand more about the relation between being one-way and having a super-core for a function $f$. Let's recall what we know in the deterministic setting: (1) if $f$ has a hard-core $b$, $f$ is not necessarily one-way because the possession of a hard-core can be due to an information loss of $f$, and (2) if $f$ is one-way, then we can construct a hard-core $b$ of $f'(x,y) = (f(x), y)$.

The first point is the same in the nondeterministic setting. The aforementioned $f(x) = x[1]$ is clearly not one-way but has a super-core $b(x) = x[-1]$, which is due to a dramatic loss of information. A more interesting question is whether the $f$ constructed in \Cref{super-bit to core} from a super-bit $g$ (we have shown $f$ possesses a super-core $b(x) = g(x)[-1]$) is one-way or not.
\begin{open problem}\label{open3.1}
If $g$ is a super-bit, is $f(x) := g(x)[1 \ldots n]$ ($n = |x|$)  a one-way function?
\end{open problem}

As for the second point, however, since being a super-core is stronger requirement than being a hard-core, it is not necessarily true that there is a ``universal'' super-core in the sense that some predicate $b$ (e.g., $b(x,y) = \langle x,y \rangle$) is a super-core for every $f'(x,y) = (f(x), y)$, where $f: \bi{n} \to \bi{m(n)}$ is a one-way function. Then, a natural question to ask is:
\begin{open problem}\label{open3.2}
Suppose $f$ is a one-way function.
If we want $\langle x,y \rangle$ to be a super-core of $f'(x,y) = (f(x), y)$,
what other properties does $f$ need to have if any?
\end{open problem}

Rather, as for the second point, we can show that it is impossible for certain $f$'s to have a super-core, even if they are possibly one-way. Such $f$'s include certain functions which are ``predominantly 1-1" infinitely often. To state this more precisely, we say $x$ is of \textbf{\textit{type 1}} if $|f^{-1}(f(x))| = 1$ and $x$ is of \textbf{\textit{type 2}} otherwise. We define $T_1(n)$ to be the set of $x \in \bi{n}$ of \textit{type 1} and $T_2(n)$ to be the set of $x \in \bi{n}$ of \textit{type 2} . Now, we establish:
\begin{theorem}\label{1-1}
Given constant integer $c 
\geq 0$ and $f: \bi{n} \to \bi{n + c}$ in \ppoly{}, if there exists infinitely many $n$'s and a polynomial $p$ such that, 
\[
\frac{|T_1|}{2^{n}} 
\geq 
\frac{2^{1+c}}{2^{1+c} + 1} + \frac{1}{p(n)},
\]
then $f$ does not have a super-core.

In particular, when $f$ is length-preserving (i.e., $c = 0$), the inequality becomes:
\[
\frac{|T_1|}{2^{n}} 
\geq 
\frac{2}{3} + \frac{1}{p(n)}.
\]
\end{theorem}
\begin{proof}
Assume $f$ has the stated property. Suppose, for a contradiction, $b$ is a super-bit of $f$. We construct as follows two algorithms $\McA_j (j = 0, 1)$, in which $\McA_0$ is co-nondeterministic and $\McA_1$ is nondeterministic, to predict $b$:
\begin{adjustwidth}{1.2cm}{1.2cm}
    Given $y \in \bi{n}$ as input, $\McA_j$ guesses $x'$ such that $f(x') = y$. If $\McA_j$ fails to guess such an $x'$, it outputs $1 - j$. Otherwise, it outputs $b(x')$.
\end{adjustwidth}

We note that if $y = f(x)$ for some $x \in T_1$, then there is always exactly one correct guess $x'$ (i.e., $x$ itself) for $\McA_j$ such that $f(x') = y = f(x)$, and whenever such an $x'$ is guessed, $b(x') = b(x)$. 
Hence, 
$\circled{1}(\McA_1) = 
\bP[\McA_1(f(x)) = b(x) = 0]
\geq
\bP[\McA_1(f(x)) = b(x) = 0 \land x \in T_1]
=
\bP[b(x) = 0 \land x \in T_1]
$,
and
$
\circled{3}(\McA_1) =
\frac{1}{2}\bP[\McA_1(y) = 1]
\geq
\frac{1}{2}\bP[y = f(x) \land b(x) = 1 \text{ for some } x \in T_1]
$.
Similarly,
$\circled{2}(\McA_0)
\geq
\bP[b(x) = 1 \land x \in T_1]
$,
and
$
\circled{4}(\McA_0)
\geq
\frac{1}{2}\bP[y = f(x) \land b(x) = 0 \text{ for some } x \in T_1]
$.

Therefore, 
$\circled{1}(\McA_1) + \circled{2}(\McA_0) + \circled{3}(\McA_1) + \circled{4}(\McA_0)
=
\pr{x \in T_1} + \frac{1}{2}\pr{y \in f(T_1)}
=
\frac{|T_1|}{2^{n}} + 
\frac{1}{2} \cdot \frac{|T_1|}{2^{n+c}}
=
\frac{2^{1+c}+1}{2^{1+c}} \cdot \frac{|T_1|}{2^{n}}
\geq
1 + \frac{1}{p(n)}
$ 
for infinitely many $n$'s, 
but this implies that
either
$\circled{1}(\McA_1) + \circled{3}(\McA_1) \geq \frac{1}{2} + \frac{1}{2p(n)}$
or
$\circled{2}(\McA_0) + \circled{4}(\McA_0) \geq \frac{1}{2} + \frac{1}{2p(n)}$.
\end{proof}
The intuitive reason that such $f$'s do not have a super-core is: such an $f$ preserves most of the information (thus, a unique pre-image can be guessed).
In the theorem, we cannot trivially relax $c$ to an arbitrary polynomial in $n$, as if we do so, the right-hand side of the inequality in the statement may exceed $1$.
This result can alternatively be proved, when $f$ is length-preserving,  by making use of other results mentioned before, as follows. Suppose, for a contradiction, $b$ is a super-core of $f$, then $g(x):= f(x)b(x)$ is a super-bit by \Cref{super-core to bit}. However, we can then easily break $g$ as follows: we witness the randomness of a given $y \in \bi{n+1}$ by guessing an $x \in \bi{n}$ such that $f(x) = y[1 \ldots n]$ and test if $b(x) = y[n+1]$; if $b(x) \neq y[n+1]$, $y$ is random.

We state a weaker but more concise corollary of \Cref{1-1}:
\begin{corollary}
If $f \in \P/\poly$ is length-preserving and 1-1 (or at least 1-1 for infinitely many $n$'s), then $f$ does not have a super-core.
\end{corollary}
This corollary contrasts the result that a length-preserving and 1-1 $f \in P/poly$ could have a hard-core (in fact, if $b$ is a hard-core of a length-preserving and 1-1 $f \in P/poly$, then $g(x):= f(x)b(x)$ is a strong PRG \cite{txt}). Thus, the corollary suggests there might be strong PRGs which are not super-bits.

Besides the above investigation on which functions can or cannot have a super-core and the relation between being a one-way function and having a super-core, another central question is:
\begin{open problem*}
What is a sensible definition of  ``nondeterministic one-way functions" if any?
\end{open problem*}
Satisfactory answers to \Cref{open3.1} and \Cref{open3.2} will shed light on this question. A possible candidate is ``a one-way function that  has a super-core", but this question  needs further exploration.

\section{Conclusion and future directions}
The current work is the first systematic investigation into nondeterministic pseudorandomness (in the cryptographic regime), building on the primitives proposed by Rudich \cite{superbit}. Our investigation reveals a fruitful area of potential directions.  
Here, we propose or summarise some intriguing future directions and open problems that we deem worth pursuing.



\begin{itemize}
    \item 
    We have provided an algorithm that achieves a sublinear-stretch for any given demi-bit.
    If we are greedier, we may wonder if we can stretch more. Specifically, if we can stretch one demi-bit to linearly many demi-bits, or to polynomially many demi-bits, or even to a pseudorandom function generator (PRFG) with exponential demi-hardness. Even if assuming we do have some $n$-demi-bits $b: \bi{n} \to \bi{2n}$, it is still unclear whether we can construct a PRFG with exponential demi-hardness by any standard algorithm for constructing PRFGs or a novel one.
    
    \item 
    Can we further refine or classify the inequalities in \Cref{sum:unpredict}? 
    Can we also characterise demi-hardness in terms of unpredictability (i.e., where is the correct place of demi-hardness in that lattice)?
    
    \item
    One-way functions are one of the most central cryptographic primitives. Then, a natural question to ask is, what is a sensible definition of nondeterministic-secure one-way functions? (A first attempt may be ``a one-way function which has a super-core", but this proposal demands further verification.)
    
    \item If $g$ is a super-bit, is $f(x) := g(x)[1...n]$ ($n = |x|$) a one-way function?
    
    \item 
    Suppose $f$ is a one-way function.
    If we want $\langle x,y \rangle$ to be a super-core of $f'(x,y) = (f(x), y)$,
    what other properties does $f$ need to have, if any?
\end{itemize}

Better answers to the questions listed above would help us to better understand the hardness of generators in the nondeterministic setting and shed light on the open problem ``whether the existence of demi-bits implies the existence of super-bits".



\appendix

\section{Demi-bits exist unless PAC-learning of small circuits is feasible}

In this section, we provide a full exposition of Pich's results from \cite{Pich}, who developed PAC-learning algorithms from breaking PRGs.
These results are important for the foundations of pseudorandomness against nondeterministic adversaries since they provide justification for the existence of demi-bits.  
\subsection{PAC-learning}
We formulate the following nonuniform version of PAC-learning:

\begin{definition}[PAC-learning, nonuniform]
A circuit class $\mathcal{C}$ is \textbf{learnable} (over the uniform distribution) by a circuit class $\mathcal{D}$ up to error $\epsilon$ with confidence $\delta$ if there is a randomized oracle family $L = \{D_n\} \in \mathcal{D}$ such that for every family $f: \bi{n} \to \bi{}$ computable by $\mathcal{C}$ and every large enough $n$, we have:
\begin{enumerate}
    \item $\Pr[w]{L^f(1^n, w)\ (1 - \epsilon)\!-\! approximates\ f} \geq \delta$, where $w$ is the random input bits to $L^f$ and the output $L^f(1^n, w)$ of $L^f$ is a string representation of an approximator $f'$ of $f$.
    $L^f(1^n, w)$ can be a description of a uniform or nonuniform algorithm of the following types: deterministic, nondeterministic, co-nondeterministic, and randomized.

    \item For every $f$ and $w$, $L^f(1^n, w)$ is $\mathcal{D}$-evaluable: there is another circuit family $E \in \mathcal{D}$ such that for every possible output $L^f(1^n, w)$ of $L$ and every $x \in \bi{n}$ given as the input to $E$, $E$ computes $L^f(1^n, w)$ on $x$. When $L^f(1^n, w)$ is (co-)nondeterministic or randomized, $E$ is allowed to be (co-)nondeterministic or randomized respectively.
\end{enumerate}
For such a learner $L$, we also say $\mathcal{C}$ is learnable by $L$ or every $C \in \mathcal{C}$ is learnable by $L$.
\end{definition}
\textbf{Remarks.}
\begin{enumerate}
    \item We say $L^f$ uses membership query if it somehow selects the set of queries made to the oracle gates. We say $L^f$ uses uniformly distributed random examples if the set of queries made is sampled uniformly at random.
    The above learning model is general enough to admit both kinds of learners and also any kind of mixture. In this work, we will only consider learners using uniformly distributed random examples. Thus by ``learning", we always mean learning using uniformly distributed random examples.

    \item The defined learning model is also general enough to admit learning over other distributions (i.e. we change the distribution of $w$ in condition (1)). When $\mathcal{D}$ contains \ppoly, learning over any polynomially generated distribution is equivalent to learning over the uniform distribution.

    \item A possible source of confusion is to leave out Condition (2). However, Condition (2) is indispensable here. The reason is that without this restriction, \ppoly{} can be efficiently learned, which is widely believed not to be the case (e.g., cf. \cite{learn_Rajgopal_Santhanam}). It is an easy exercise to prove the learnability of \ppoly{} without the second condition by applying the Occam's Razor theorem established in \cite{OccamsRazor}. Intuitively, a learner can learn \ppoly{} efficiently by remembering the samples and postponing all the ``learning" to the hypothesis evaluation stage.
    
    \item The confidence $\delta$ and accuracy $1 - \epsilon$ of a learner in $\McD$ can be efficiently boosted to constants less than $1$ in standard ways (cf. \cite{learning_txt}) when they are not negligible with respect to $\McD$. For example, when $\McD = P/poly$, it is sufficient for $\delta$ and $1 - \epsilon$ to achieve $p(n)$ and $1/2 + q(n)$ respectively for any polynomials $p$ and $q$. Beyond the remark here, boosting will be a digression from the theme of this work.
\end{enumerate}

\subsection{Learning based on nonexistence assumptions}

We recall a construction from \cite{crypto}: for a positive integer $m$ and a circuit $C: \bi{n} \to \bi{}$, define generator $G_{m,C}: \bi{mn} \to \bi{mn+m}$, which maps $m$ $n$-bit strings $x_1, ..., x_m$ to $x_1, C(x_1), ... , x_m, C(x_m)$.

We reformulate Theorem 7 on average-case learning in \cite{crypto} into our PAC-learning framework as the following lemma:
\begin{lemma} \label{lemma:learn}
Given a circuit class $\mathcal{C}$, if there is an $m$ and an $s(n)$-size circuit $D$ such that for every $C \in \mathcal{C}$
\[ \pr{D(y)=1} - \pr{D(G_C(x))=1} \geq 1/s \text{, where } G_C = G_{m,C} \text{,} \]
then there is a randomized polynomial time (in $n$ and $|\langle D \rangle|$, where $\langle D \rangle$ is a given string representation of $D$) algorithm $L$ that learns $\mathcal{C}$ with confidence $1/2m^2s$ up to error $1/2 - 1/2ms$.

In particular, if D is a nondeterministic or co-nondeterministic circuit, the output of L is allowed to be a nondeterministic or co-nondeterministic algorithm.
\end{lemma}
\begin{proof}
Given any $C \in \mathcal{C}$, $L$ randomly chooses an $i \in [m]$, bits $r_1, ..., r_m$, and n-bit strings $x_1, ..., x_m$ except $x_i$,
queries $C$ on $x_1, ..., x_{i-1}$ to get $C(x_1), ..., C(x_{i-1})$,
and outputs $C'$, which predicts $C$ as follows:
given any $n$-bit input $x_i$,
$C'$ emulates $D$ on $(x_1, C(x_1), ..., x_{i-1}, C(x_{i-1}), x_i, r_i, ..., x_m, r_m)$ to get an output bit $p_i$. If $p_i = 1$, C' outputs $\bar{r_i}$; if $p_i = 0$, C' outputs $r_i$.

We next want to prove that $L$ indeed learns $C$ in the desired sense.
Given random bits $r_1, ..., r_m$ and random $n$-bit strings $x_1, ..., x_m$, we define $p_i := D(x_1, C(x_1), ..., x_{i-1}, C(x_{i-1}), x_i, r_i, ..., x_m, r_m)$. Then (note: here we are applying a hybrid argument),
\begin{align*}
    1/s
    &\leq \pr{D(x)=1}- \pr{D(G_C(x))=1} \\
    &= \pr{p_1 = 1} - \pr{p_m = 1} \\
    &= \sum_{i=1}^{m-1} (\pr{p_i = 1} - \pr{p_{i+1} = 1})
\end{align*}
Hence, there exist i such that $\pr{p_i = 1} - \pr{p_{i+1} = 1} \geq 1/ms$, and therefore the $i$ $L$ chooses satisfies $\pr{p_i = 1} - \pr{p_{i+1} = 1} \geq 1/ms$ with probability $\geq 1/m$. When $L$ has successfully chosen such \mbox{an $i$,}
\begin{align*}
    &\ \underset{\substack{r_1, ..., r_m\\x_1, ..., x_m}}{\mathbb{P}}
       [C'(x_i) = C(x_i)] \\
    &= \pr{p_i = 1, r_i \ne C(x_i)} +
       \pr{p_i = 0, r_i = C(x_i)}\\
    &= \frac{1}{2} \pr{p_i = 1 | r_i \ne C(x_i)} +
       \frac{1}{2} \pr{p_i = 0 | r_i = C(x_i)}\\
    &= \frac{1}{2} \pr{p_i = 1 | r_i \ne C(x_i)} +
       \frac{1}{2} (1 - \pr{p_i = 1 | r_i = C(x_i)})\\
    &= \frac{1}{2} +
       \frac{1}{2} \pr{p_i = 1 | r_i \ne C(x_i)} -
       \frac{1}{2} \pr{p_i = 1 | r_i = C(x_i)} \\
    &= \frac{1}{2} +
       \frac{1}{2} \pr{p_i = 1 | r_i \ne C(x_i)} +
       (\frac{1}{2} \pr{p_i = 1 | r_i = C(x_i)} -
       \pr{p_i = 1 | r_i = C(x_i)}) \\
    &= \frac{1}{2} +
       (\pr{p_i = 1, r_i \ne C(x_i)} +
       \pr{p_i = 1, r_i = C(x_i)}) -
       \pr{p_i = 1 | r_i = C(x_i)} \\
    &= \frac{1}{2} + \pr{p_i = 1} - \pr{p_{i+1} = 1} \\
    &\geq \frac{1}{2} + \frac{1}{ms}
\end{align*}
Let $p =
\underset{\substack{r_1, ..., r_m\\x_1, ..., x_m \, except \, x_i}}{\mathbb{P}}
[\underset{x_i}{\mathbb{P}}[C'(x_i) = C(x_i)] >= 1/2 + 1/2ms]$. As there exist $0 \leq a < 1/2$ such that $1/2 + 1/2ms + a$ represents the average accuracy of the predictor $C'$ when $r_1, ..., r_m, x_1, ..., x_{i-1}, x_{i+1}, ..., x_m$ satisfy $\underset{x_i}{\mathbb{P}}[C'(x_i) = C(x_i)] >= 1/2 + 1/2ms$, we have
\begin{displaymath}
p\left(\frac{1}{2} + \frac{1}{2ms} + a\right) +
(1-p)\left(\frac{1}{2} + \frac{1}{2ms}\right) \geq
\underset{\substack{r_1, ..., r_m\\x_1, ..., x_m}}{\mathbb{P}}[C'(x_i) = C(x_i)]
\geq \frac{1}{2} + \frac{1}{ms},
\end{displaymath}
and thus
\[p \geq \frac{1}{2ms}\cdot\frac{1}{a} \geq \frac{1}{ms}. \]
Therefore, $L$ outputs a $(1/2 + 1/2ms)$-accurate $C'$ with confidence $1/m^2s$.
\end{proof}
\begin{comment*} When $D$ is nondeterministic: (1) if $r_i = 0$, then the $C'$ learned by $L^C$ is also nondeterministic; (2) if $r_i = 0$, then the $C'$ learned by $L^C$ is co-nondeterministic.
\end{comment*}

In the remaining of this section, we are going to derive learning algorithms based on the assumptions of the nonexistence of i.o. demi-bits and demi-bits respectively. It is obvious that the non-existence of i.o. demi-bits is a stronger assumption than the non-existence of demi-bits. We will also see that, given our proof strategy, the stronger assumption yields a better learning result, in a sense that will be clear later.

A proof of the following theorem was  sketch in \cite{Pich}:
\begin{theorem} \label{thm:learn}
Assume the nonexistence of i.o.~demi-bits. Then for  every $c \in \bN$, $Circuit[n^c]$ is learnable by $Circuit[2^{n^{o(1)}}]$ (by taking random examples) with confidence $1/2^{n^{o(1)}}$ up to error $1/2 - 1/2^{n^{o(1)}}$, where the learner is allowed to generate a nondeterministic or co-nondeterministic algorithm approximating the target function.
\end{theorem}
\begin{proof}
Assume the nonexistence of i.o.~demi-bits and $c \in \bN$. There is an encoding scheme for $Circuit[n^c]$ such that every $C \in Circuit[n^c]$ with $n$-bit input can be encoded as a string $\langle C \rangle$ of length $\leq d = d(n)$, where $d(n)$ is a polynomial. Set $m = n^d + 1$ and consider a generator $G: \bi{mn + n^d} \to \bi{mn+m}$, which interprets the last $n^d$ input bits as a description of some $C \in Circuit[n^c]$ and then computes on the remaining $mn$ bits of input as $G_{m,C}$. We note that $G$ is in \ppoly  (whenever the encoding scheme for $Circuit[n^c]$ is reasonable). Because $G$ is not an i.o. demi-bit, there is nondeterministic circuit $D$ of sub-exponential size such that for all sufficiently large $n$'s
\[
\pr{D(y) = 1} \geq 1/|D|
\ \text{ and }\
\pr{D(G(x)) = 1} = 0.
\]
In particular, for every $C \in Circuit[n^c]$,
$\pr{D(G(\langle C \rangle, x) = 1} = 0$. That is $\pr{D(G_{m,C}(x))=1} = 0$ for every $C \in Circuit[n^c]$. Therefore,
\[\pr{D(y)=1} - \pr{D(G_{m,C}(x))=1} \geq 1/|D|.\]
Because $m = n^d + 1$ is polynomial in $n$ and $|D|$ is sub-exponential, by \Cref{lemma:learn}, $\mathcal{C}$ can be learned by a randomized circuit family of size $2^{n^{o(1)}}$ with confidence $1/2^{n^{o(1)}}$ up to error $1/2 - 1/2^{n^{o(1)}}$.
\end{proof}

If we weaken the assumption to the non-existence of demi-bits, with the same proof, we are able to learn $Circuit[n^c]$ in a weaker sense formulated as below:
\begin{theorem}
Assume the nonexistence of demi-bits. Then for every $c \in \bN$, there is an infinite monotone sequence $\{n_i\} \subseteq \bN$ such that $Circuit[n^c]$ is learnable by $Circuit[2^{n^{o(1)}}]$ (by taking random examples) with confidence $1/2^{n^{o(1)}}$ up to error $1/2 - 1/2^{n^{o(1)}}$ for every $n \in \{n_i\}$, where the learner is allowed to generate a nondeterministic or co-nondeterministic algorithm approximating the target function.
\end{theorem}

Because the existence of $N\tilde{P}/qpoly$-natural proofs (a weaker assumption than the existence assumption of \nppoly-natural proofs) rules out the existence of i.o.~super-bits \cite{superbit}, by \Cref{thm:learn}, we have:
\begin{corollary}
Assume the existence of i.o.~demi-bits implies the existence of i.o. super-bits. 
If there exists an $N\tilde{P}/qpoly$-natural property useful against \ppoly{}, then for  every $c \in \bN$, $Circuit[n^c]$ is learnable by $Circuit[2^{n^{o(1)}}]$ (by taking random examples) with confidence $1/2^{n^{o(1)}}$ up to error $1/2 - 1/2^{n^{o(1)}}$, where the learner is allowed to generate a nondeterministic or co-nondeterministic algorithm approximating the target function.
\end{corollary}

\section*{Acknowledgments}
We are indebted to Jan Pich who suggested looking at demi-bits and provided many clarifications regarding his work \cite{Pich}. We are grateful to Rahul Santhanam for very useful discussions  and specifically mentioning the  potential application appearing in \Cref{thm:ac}. Finally, we  wish to thank  Hanlin Ren for very useful comments on the manuscript as well as  Oliver Korten and Yufeng Li for further discussions.


\bibliographystyle{alpha}
\bibliography{Demi-bib}


\end{document}